\documentclass[pra,
superscriptaddress,
preprintnumbers,amsmath,amssymb]{revtex4}
\usepackage{graphicx}% Include figure files
\usepackage{amsthm}
\usepackage{color}
\usepackage[all]{xy}
\usepackage{braket}

\newtheorem{Thm}{Theorem}

\theoremstyle{definition}
\newtheorem{Def}[Thm]{Definition}

\begin{document}

\title{Hybrid quantum linear equation algorithm \\
and its experimental test on IBM Quantum Experience}

\author{Yonghae Lee}%\email{yonghaelee@khu.ac.kr}
\affiliation{
Department of Mathematics and Research Institute for Basic Sciences,
Kyung Hee University, Seoul 02447, Korea
}
\author{Jaewoo Joo}
\affiliation{
School of Computational Sciences, Korea Institute for Advanced Study, Seoul 02455, Korea}
\affiliation{
Clarendon Laboratory, University of Oxford, Parks Road, Oxford OX1 3PU, United Kingdom}
\author{Soojoon Lee}%\email{level@khu.ac.kr}
\affiliation{
Department of Mathematics and Research Institute for Basic Sciences,
Kyung Hee University, Seoul 02447, Korea
}
\affiliation{
School of Mathematical Sciences and Centre for the Mathematics
and Theoretical Physics of Quantum Non-Equilibrium Systems,
University of Nottingham, University Park, Nottingham NG7 2RD, United Kingdom
}

\date{\today}

%%%%%%%%%%%%%%%%%%%%%%%%%%%%%%%%%%%%%%%%%%%%%%%%%%%%%%%%%%%%%%%%%%%%%%%%%%%%%%%%%%%%%%%%%%%%%%%%%
%%%                                                                                                                           
%%%   Abstract
%%%                                                                                                                           
%%%%%%%%%%%%%%%%%%%%%%%%%%%%%%%%%%%%%%%%%%%%%%%%%%%%%%%%%%%%%%%%%%%%%%%%%%%%%%%%%%%%%%%%%%%%%%%%%
\begin{abstract}
We propose a hybrid quantum algorithm based on the Harrow-Hassidim-Lloyd (HHL) algorithm for solving a system of linear equations. In our hybrid scheme, a classical information feed-forward is required from the quantum phase estimation algorithm to reduce a circuit depth from the original HHL algorithm. 
In this paper,
we show that this hybrid algorithm is functionally identical to the HHL algorithm
under the assumption that the number of qubits used in algorithms is large enough.
In addition,
it is experimentally examined with four qubits in the IBM Quantum Experience setups,
and the experimental results of our algorithm
show higher accurate performance on specific systems of linear equations
than that of the HHL algorithm.
\end{abstract}

\maketitle
%%%%%%%%%%%%%%%%%%%%%%%%%%%%%%%%%%%%%%%%%%%%%%%%%%%%%%%%%%%%%%%%%%%%%%%%%%%%%%%%%%%%%%%%%%%%%%%%%
%%%                                                                                                                           
%%%   Introduction
%%%                                                                                                                           
%%%%%%%%%%%%%%%%%%%%%%%%%%%%%%%%%%%%%%%%%%%%%%%%%%%%%%%%%%%%%%%%%%%%%%%%%%%%%%%%%%%%%%%%%%%%%%%%%
\section{Introduction}
A quantum computer is a physical machine based on quantum physics.
Since the Shor's algorithm was known to be a method
for factoring a very large number with exponential speed-up on a quantum computer~\cite{S97},
various quantum algorithms have been theoretically introduced under the assumption of noiseless quantum computers. However, the performance of quantum algorithms in practice suffers from physical errors in noisy quantum devices under technical limitations (e.g., decoherence).
Thus,
it is of great importance to find more efficient and error-robust methods
for existing quantum algorithms within physical error thresholds for near-term future applications.

The Harrow-Hassidim-Lloyd (HHL) algorithm~\cite{HHL09} is a well-known and quantum algorithm
for finding the solution $\vec{x}$ of a given system of linear equations
represented by an input matrix $\hat{A}$ and a vector $\vec{b}$.
Intuitively,
the HHL algorithm performs the inverse of the matrix $\hat{A}$ on the vector $\vec{b}$ in a heralded way and is more efficiently operated with sparse matrix $\hat{A}$.
Because the HHL algorithm demonstrates how to use quantum computers for fundamental mathematical problems, it provides important impact on other quantum applications in other quantum applications such as the quantum machine learning algorithm~\cite{LGZ16}
and the high-order quantum algorithm~\cite{B14} for solving differential equations.

The purpose of this paper is to provide a modified version
of the \emph{original} HHL algorithm~\cite{HHL09} to be efficiently operated on both classical and quantum computers in sequential steps.
The main idea of our hybrid algorithm is
to remove an unnecessary quantum part of the original HHL algorithm
with prior classical information, so we call it the \emph{hybrid} HHL algorithm.
This makes the shortened circuit depth of the original algorithm 
without losing quantum advantages dependent to the original algorithm.
We also demonstrate the hybrid HHL algorithm compared with the original one
with different eigenvalues of $\hat{A}$ in the IBM Quantum eXperience (IBMQX) setups,
and show that our hybrid algorithm has more enhanced performance than the other.

This paper is organized as follows.
In Sec.~\ref{sec:Prel}
we introduce the definitions and details of the original HHL algorithm theoretically.
In Sec.~\ref{sec:hybridHHL},
we describe the hybrid HHL algorithm coped with our specific linear system. 
In Sec.~\ref{sec:CIaE},
we verify that the hybrid algorithm has reduced the effect of the errors from the HHL algorithm tested on the IBMQX setups.
Finally, in Sec.~\ref{sec:Conclusion}, we make a summary of our results and a further discussion.

%%%%%%%%%%%%%%%%%%%%%%%%%%%%%%%%%%%%%%%%%%%%%%%%%%%%%%%%%%%%%%%%%%%%%%%%%%%%%%%%%%%%%%%%%%%%%%%%%
%%%                                                                                                                           
%%%   Preliminaries
%%%                                                                                                                           
%%%%%%%%%%%%%%%%%%%%%%%%%%%%%%%%%%%%%%%%%%%%%%%%%%%%%%%%%%%%%%%%%%%%%%%%%%%%%%%%%%%%%%%%%%%%%%%%%
\section{Preliminaries} \label{sec:Prel}

%%%%%%%%%%%%%%%%%%%%%%%%%%%%%%%%%%%%%%%%%%%%%%%%%%%%%%%%%%%%%%%%%%%%%%%%%%%%%%%%%%%%%%%%%%%%%%%%%%
%%%   Definitions
%%%%%%%%%%%%%%%%%%%%%%%%%%%%%%%%%%%%%%%%%%%%%%%%%%%%%%%%%%%%%%%%%%%%%%%%%%%%%%%%%%%%%%%%%%%%%%%%%%
\subsection{Definitions}
A general form of linear systems of equations is given in
\begin{equation} \label{eq:linear_system_of_equations}
\hat{A}\vec{x}=\vec{b},
\end{equation}
where $\hat{A}$ is a matrix and $\vec{b}$ is a vector.
Throughout this paper,
it is assumed that the matrix $\hat{A}$ is Hermitian and the vector $\vec{b}$ is unit.
Then the matrix $\hat{A}$ has a spectral decomposition~\cite{W13}
\begin{equation} \label{eq:spectral_decomp}
\hat{A}=\sum_{j=1}^{l}\lambda_j\ket{u_j}\bra{u_j},
\end{equation}
where
$\lambda_j$ is an eigenvalue of $\hat{A}$ corresponding to the eigenstate $\ket{u_j}$.
From this decomposition,
a unitary operator $U_{\hat{A}}$ 
is defined as follows:
\begin{equation} \label{eq:unitary_from_A}
U_{\hat{A}}=e^{2\pi i \hat{A}}=\sum_{j=1}^{l} e^{2\pi i\lambda_j}\ket{u_j}\bra{u_j}.
\end{equation}
It is easy to see that
for any non-zero eigenvalue $\lambda_j$ of $\hat{A}$
there exists $\lambda'_j\in(0,1)$ such that $e^{2\pi i\lambda'_j}=e^{2\pi i\lambda_j}$.
Thus,
for convenience,
we may assume that the eigenvalues of $\hat{A}$ are in $(0,1)$.

We then introduce three definitions to explain the main idea of this work.
\begin{Def} \label{def:nBE}
Let $\lambda$ be a positive real number with the range of $(0,1)$,
then its \emph{binary representation} is given by
\begin{equation}
\lambda=0.b_1b_2b_3\cdots,
\end{equation} 
where $b_k\in\{0,1\}$ is a $k$-th bit of the binary representation.
For $n\in\mathbb{N}$,
the \emph{$n$-binary estimation of $\lambda$},
say $\lambda(n)$,
is given by
\begin{equation}
\lambda(n):=b_1b_2b_3\cdots b_n \approx 2^{n} \lambda.
\label{tilde_lambda}
\end{equation}
\end{Def}

\begin{Def} \label{def:eigenmean}
Let $\{\lambda_j\}_{j=1}^{l}$ be the set of all non-zero eigenvalues
of a Hermitian matrix $\hat{A}$.
For $k\in\mathbb{N}$,
define a constant $\bar{m}_k$ as
\begin{equation}
\bar{m}_k:={1 \over l} \left(\sum_{j=1}^{l}b_k^j\right),
\end{equation}
where $b_k^j$ is the $k$-th bit of the binary representation of $\lambda_j$.
We call $\bar{m}_k$ the \emph{$k$-th eigenmean of $\hat{A}$}.
Moreover,
if $\bar{m}_k=0$ or $1$, $\bar{m}_k$ is called \emph{fixed}.
\end{Def}

In Definition~\ref{def:eigenmean},
we remark that
if the $k$-th eigenmean of $\hat{A}$ is fixed
then every $k$-th bits of the binary representations of $\{\lambda_j\}_{j=1}^{l}$
is equal,
that is,
if $\bar{m}_k$ is fixed then $b_k^{j_1}=b_k^{j_2}$ for any $1\le j_1,j_2\le l$.

\begin{Def}
Let $\lambda$ be an eigenvalue of a Hermitian matrix $\hat{A}$ and let $n\in\mathbb{N}$.

(i) $\lambda$ is called \emph{perfectly $n$-estimated},
if $\lambda$ satisfies $2^n\lambda =\lambda(n)$,
where $\lambda(n)$ is the $n$-binary estimation of $\lambda$ in Definition~\ref{def:nBE}.

(ii) The matrix $\hat{A}$ is called \emph{perfectly $n$-estimated},
if all the eigenvalues of $\hat{A}$ are perfectly $n$-estimated.
\end{Def}

%%%%%%%%%%%%%%%%%%%%%%%%%%%%%%%%%%%%%%%%%%%%%%%%%%%%%%%%%%%%%%%%%%%%%%%%%%%%%%%%%%%%%%%%%%%%%%%%%%
%%%   Harrow-Hassidim-Lloyd algorithm
%%%%%%%%%%%%%%%%%%%%%%%%%%%%%%%%%%%%%%%%%%%%%%%%%%%%%%%%%%%%%%%%%%%%%%%%%%%%%%%%%%%%%%%%%%%%%%%%%%
\subsection{HHL algorithm} \label{sec:HHL}
For a given $\hat{A}\vec{x}=\vec{b}$,
the HHL algorithm~\cite{HHL09} was devised
to figure out the approximation of the expectation value $\vec{x}^{\dagger}M\vec{x}$
for some operator $M$.
In the algorithm,
$\vec{b}$ is represented as a quantum state
$\ket{b}_V=\sum_{j=1}^{l}\alpha_j\ket{u_j}_V$,
where $\ket{u_j}_V$ is an eigenstate of $\hat{A}$
and $\alpha_j\in\mathbb{C}$ such that $\sum_{j=1}^{l}|\alpha_j|^2=1$,
and the solution is given as a quantum state 
\begin{equation} \label{eq:normalized_sol}
\ket{x}_V=\frac{{\hat{A}}^{-1}\ket{b}_V}{\|{\hat{A}}^{-1}\ket{b}_V\|},
\end{equation}
where ${\hat{A}}^{-1}$ is the inverse matrix of $\hat{A}$.
As shown in Fig.~\ref{fig:algorithm_HHL},
the HHL algorithm with $n$-qubit register consists of three main parts:
the quantum phase estimation (QPE) algorithm~\cite{K95,MP95,DBE98,HMPEPC97}
without the final measurement part (we call it as QPE part),
the ancilla quantum encoding (AQE) part,
in which the ancillary qubit $A$ conditionally operates on the state of the register qubits,
and the inverse QPE part.

\begin{figure}
\centering
\includegraphics[width=.6\linewidth,trim=0cm 0cm 0cm 0cm]{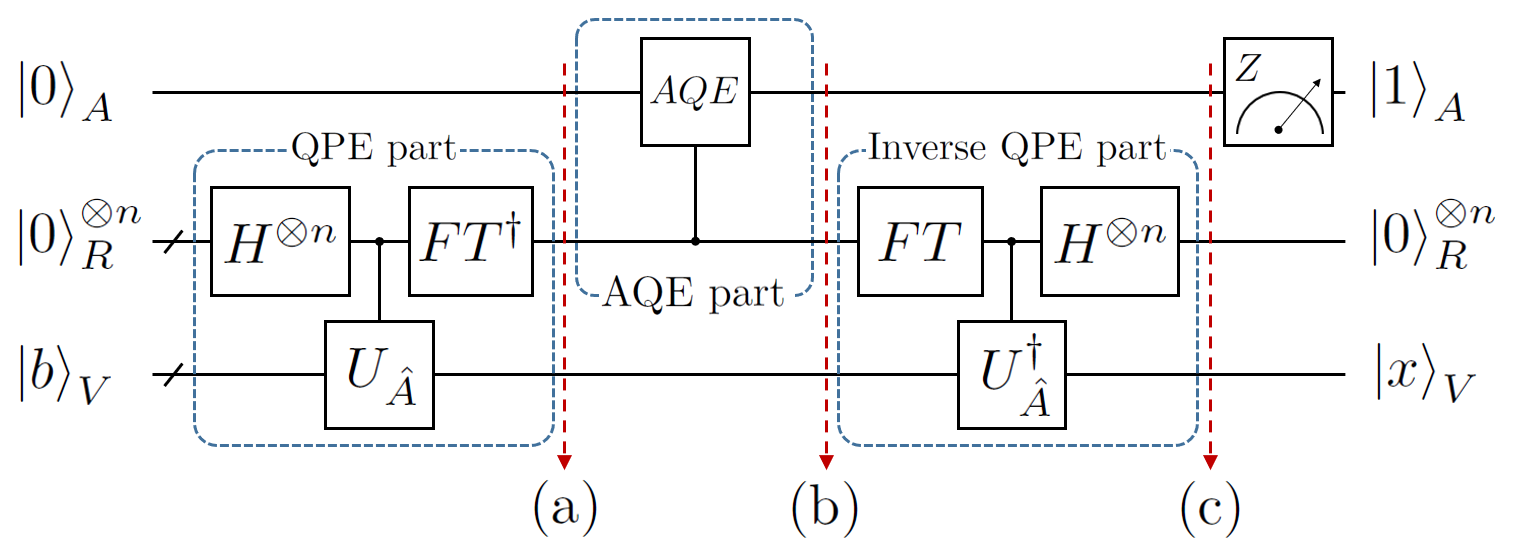}
\caption{
The circuit diagram for the HHL algorithm~\cite{HHL09}:
the circuit consists of the QPE part,
the AQE part
and the inverse QPE part.
The unitary gate $U_{\hat{A}}=e^{2\pi i \hat{A}}$ is used for controlled-unitary gates
between the register $R$ and the input qubit $V$
while the controlled $AQE$ indicates a set of controlled gates
between the register $R$ and the ancillary qubit $A$.
}
\label{fig:algorithm_HHL}
\end{figure}

%%%%%%%%%%%%%%%%%%%%%%%%%%%%%%%%%%%%%%%%%%%%%%%%%%%%%%%%%%%%%%%%%%%%%%%%%%%%%%%%%%%%%%%%%%%%%%%%%%
%%%   QPE part
%%%%%%%%%%%%%%%%%%%%%%%%%%%%%%%%%%%%%%%%%%%%%%%%%%%%%%%%%%%%%%%%%%%%%%%%%%%%%%%%%%%%%%%%%%%%%%%%%%
We first describe the QPE part of the HHL algorithm and assume that the initial state 
is prepared in $\ket{0}_A\otimes\ket{0}_R^{\otimes n}\otimes\ket{b}_V$ with a $n$-qubit register.
After finishing the QPE part,
the state at step (a) in Fig.~\ref{fig:algorithm_HHL} is written
by the superposition of the state (see details in Eq.~(\ref{1st_QPEA}))
with index $j$ and the ancillary qubit $\ket{0}_A$ such that
\begin{equation} \label{eq:after_PE_part}
\ket{0}_A\otimes\sum_{j=1}^{l}\sum_{x=0}^{2^n-1}\alpha_j\beta_{x|j}\ket{x}_R\otimes\ket{u_j}_V,
\end{equation}
where $\beta_{x|j}=\frac{1}{2^n}\sum_{y=0}^{2^n-1} e^{2\pi i y\left(\lambda_j-{x}/{2^n}\right)}$.
Then
the estimated value $x$ in Eq.~(\ref{eq:after_PE_part}) can be relabeled
with ${\lambda_x}=x/2^n$ such as
\begin{equation} \label{eq:after_PE_part_relabel}
\ket{0}_A
\otimes\sum_{j=1}^{l}\sum_{x=0}^{2^n-1}\alpha_j\beta_{x|j}\ket{{\lambda_x}}_R
\otimes\ket{u_j}_V.
\end{equation}

%%%%%%%%%%%%%%%%%%%%%%%%%%%%%%%%%%%%%%%%%%%%%%%%%%%%%%%%%%%%%%%%%%%%%%%%%%%%%%%%%%%%%%%%%%%%%%%%%%
%%%   AQE part
%%%%%%%%%%%%%%%%%%%%%%%%%%%%%%%%%%%%%%%%%%%%%%%%%%%%%%%%%%%%%%%%%%%%%%%%%%%%%%%%%%%%%%%%%%%%%%%%%%
In the AQE part,
a quantum encoding operation about the ancillary qubit $A$ is performed,
and the the controlled $AQE$ in Fig.~\ref{fig:algorithm_HHL} is given by
\begin{equation} \label{eq:n_reg_ancillary_quantum_encoding_mapping}
\ket{0}_A
\otimes\ket{{\lambda_x}}_R
\mapsto
\left(\sqrt{1-\frac{c^2}{{\lambda_x}^2}}\ket{0}_A+\frac{c}{{\lambda_x}}\ket{1}_A\right)
\otimes\ket{{\lambda_x}}_R,
\end{equation}
where $c=1/\|{\hat{A}}^{-1}\ket{b}\|$.
In practice,
the value $c$ in Eq.~(\ref{eq:n_reg_ancillary_quantum_encoding_mapping})
has to be chosen with $O(1/\kappa)$ as in~\cite{HHL09},
where $\kappa$ is called the condition number of $\hat{A}$.
Then
the state at step (b) in Fig.~\ref{fig:algorithm_HHL} is equal to
\begin{equation} \label{eq:AQE_state}
\sum_{j=1}^{l}\sum_{x=0}^{2^n-1}
\left(\sqrt{1-\frac{c^2}{{{\lambda_x}}^2}}\ket{0}_A+\frac{c}{{\lambda_x}}\ket{1}_A\right)
\otimes\alpha_j\beta_{x|j}\ket{{\lambda_x}}_R\otimes\ket{u_j}_V.
\end{equation}

%%%%%%%%%%%%%%%%%%%%%%%%%%%%%%%%%%%%%%%%%%%%%%%%%%%%%%%%%%%%%%%%%%%%%%%%%%%%%%%%%%%%%%%%%%%%%%%%%%
%%%   Inverse QPE part
%%%%%%%%%%%%%%%%%%%%%%%%%%%%%%%%%%%%%%%%%%%%%%%%%%%%%%%%%%%%%%%%%%%%%%%%%%%%%%%%%%%%%%%%%%%%%%%%%%
If all the eigenvalues $\lambda_j$ are perfectly $n$-estimated
then $\beta_{x|j} = \delta_{x,2^n\lambda_j}$,
and the state in Eq.~(\ref{eq:AQE_state}) becomes
\begin{equation}
\sum_{j=1}^{l}
\left(\sqrt{1-\frac{c^2}{{\lambda_j}^2}}\ket{0}_A+\frac{c}{\lambda_j}\ket{1}_A\right)
\otimes\alpha_j\ket{\lambda_j}_R
\otimes\ket{u_j}_V.
\end{equation}
Then after performing the inverse QPE part,
the state at step (c) in Fig.~\ref{fig:algorithm_HHL} is represented as
\begin{equation} \label{eq:after_IQPE}
\sum_{j=1}^{l}
\left(\sqrt{1-\frac{c^2}{{\lambda_j}^2}}\ket{0}_A+\frac{c}{\lambda_j}\ket{1}_A\right)
\otimes\ket{0}_R^{\otimes n}
\otimes\alpha_j\ket{u_j}_V,
\end{equation}
in which all the register qubits are reseted in $\ket{0}_R^{\otimes n}$.
The normalized solution of the linear equation appears
when the measurement of the ancillary qubit $A$ is performed in $Z$-axis.
In other words,
if the outcome state of $A$ is $\ket{1}_A$,
the state describing the qubit system $V$ successfully represents
the solution of the linear equation as follows:
\begin{equation} \label{eq:algorithm_HHL_final_state}
\frac{1}{\|{\hat{A}}^{-1}\ket{b}\|}\sum_{j=1}^{l}\frac{\alpha_j}{\lambda_j}\ket{u_j}_V,
\end{equation}
where $\|{\hat{A}}^{-1}\ket{b}\|^2=\sum_{j=1}^{l} |\alpha_j|^2/\lambda_j^2$.
On the other hand,
if there exists an eigenvalue of $\hat{A}$ which is not perfectly $n$-estimated,
then the total state of Eq.~(\ref{eq:after_IQPE}) becomes a pure entangled state.
In this case,
the pure state in Eq.~(\ref{eq:algorithm_HHL_final_state}) turns into
a mixed state which is not the solution of the linear equation.

%%%%%%%%%%%%%%%%%%%%%%%%%%%%%%%%%%%%%%%%%%%%%%%%%%%%%%%%%%%%%%%%%%%%%%%%%%%%%%%%%%%%%%%%%%%%%%%%%%
%%%                                                                                                                           
%%%   Hybrid HHL algorithm
%%%                                                                                                                           
%%%%%%%%%%%%%%%%%%%%%%%%%%%%%%%%%%%%%%%%%%%%%%%%%%%%%%%%%%%%%%%%%%%%%%%%%%%%%%%%%%%%%%%%%%%%%%%%%%
\section{Hybrid HHL algorithm} \label{sec:hybridHHL}
\begin{figure}[t]
\begin{minipage}[b]{.49\linewidth}
\includegraphics[width=\linewidth,trim=0cm 0cm 0cm 0cm]{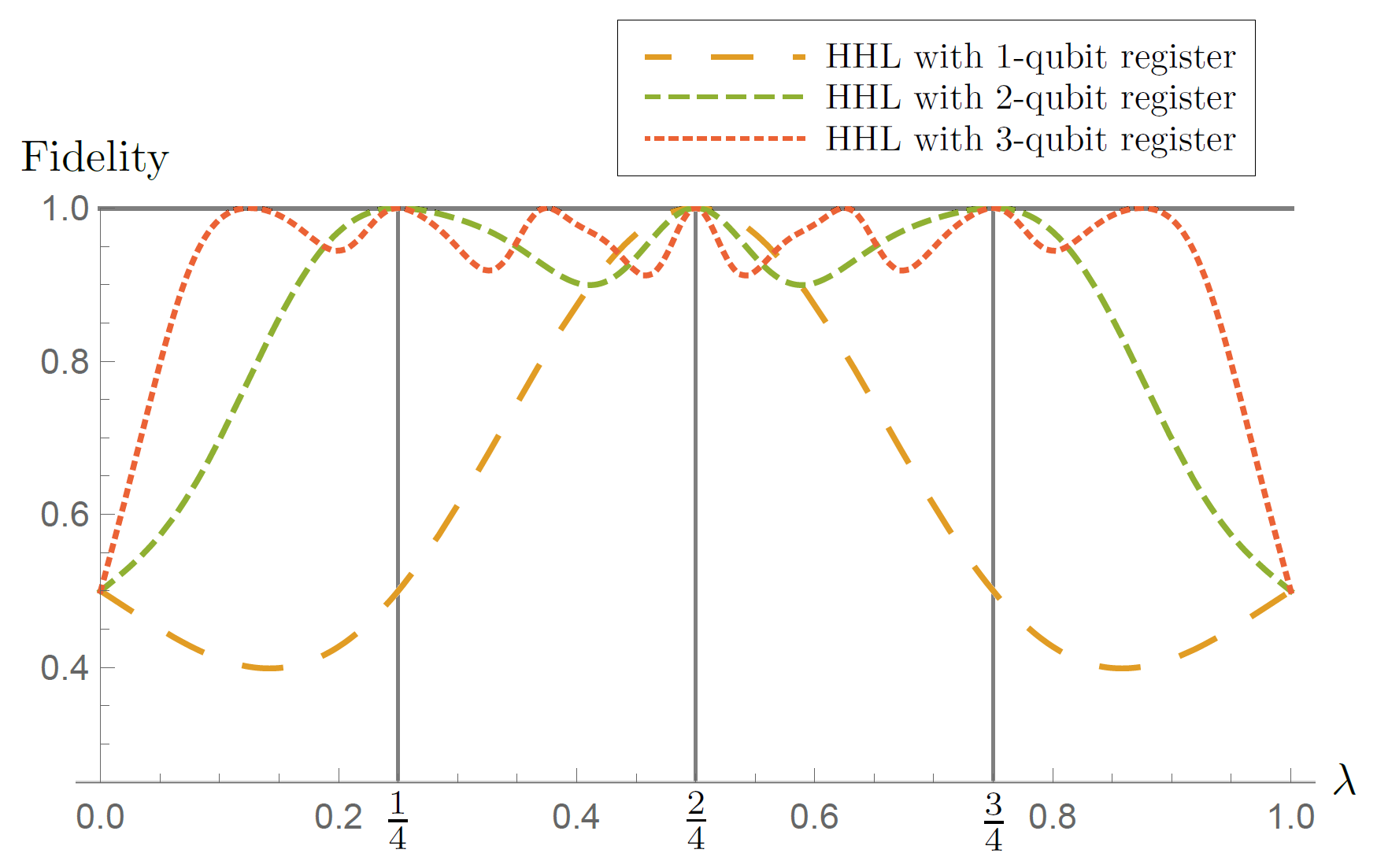}
\center{(a)}
\end{minipage}
\begin{minipage}[b]{.49\linewidth}
\includegraphics[width=\linewidth,trim=0cm 0cm 0cm 0cm]{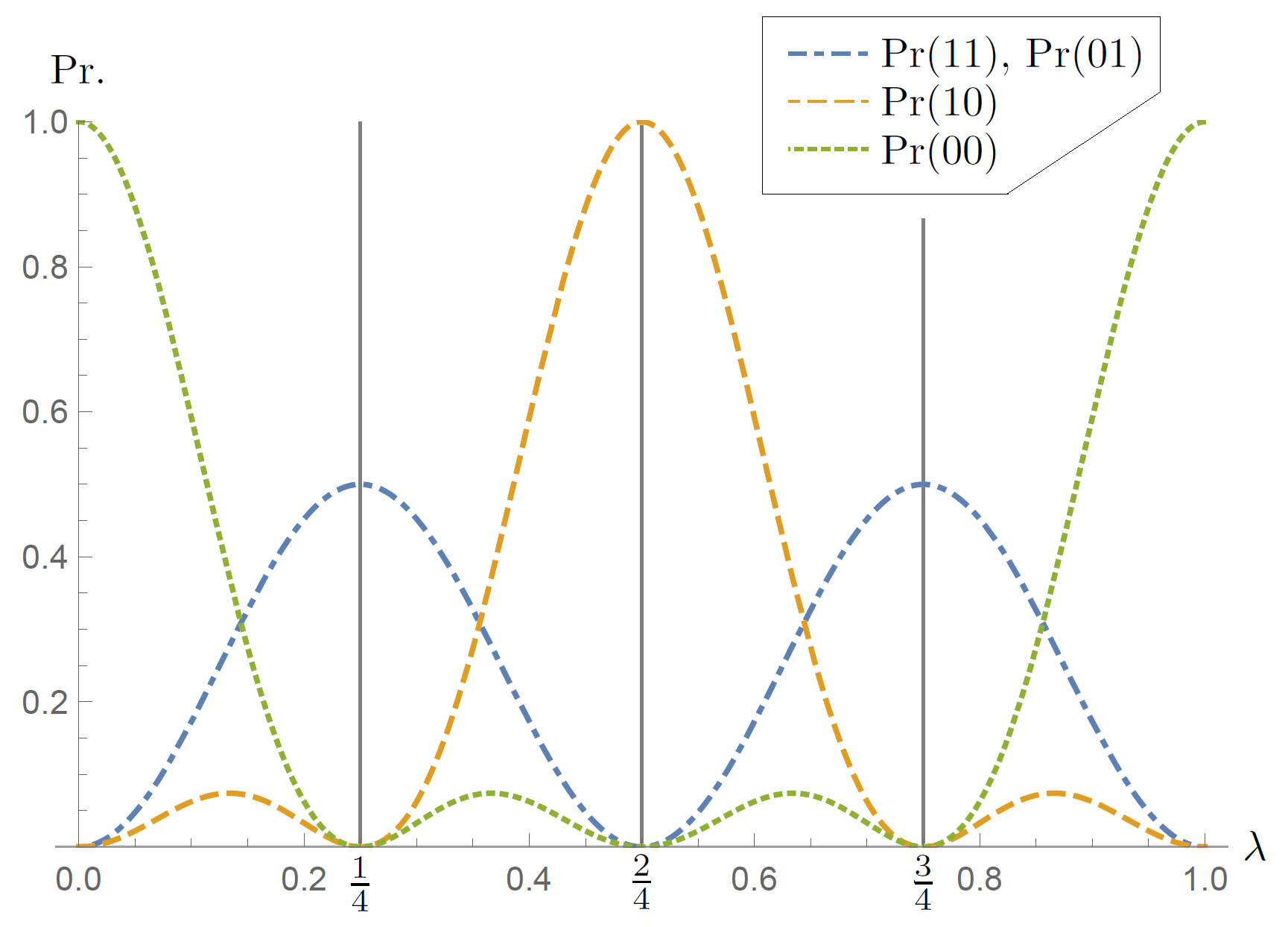}
\center{(b)}
\end{minipage}
\caption{
(a) Fidelities between solutions of the linear systems of equations in Eq.~(\ref{eq:specific_form2})
and output states obtained from the HHL algorithm with $k$-qubit register for $k=1,2,3$. (b)
The probability distribution for measurement outcomes:
The QPEA with 2-qubit register is performed on ${\hat{A}}_\lambda$ and $\vec{b}$
in Eq.~(\ref{eq:specific_form2}).
Since the QPEA makes use of 2 qubits as register,
its measurement outcomes are two-bit strings $b_1b_2$ with $b_1,b_2\in\{0,1\}$,
and $\Pr(b_1b_2)$ denotes the probability that the outcome is $b_1b_2$.
Details of the probabilities are presented in Eqs.~(\ref{Prob1}) and (\ref{Prob2}).
}
\label{fig:graph_fidelity}
\end{figure}
%%%%%%%%%%%%%%%%%%%%%%%%%%%%%%%%%%%%%%%%%%%%%%%%%%%%%%%%%%%%%%%%%%%%%%%%%%%%%%%%%%%%%%%%%%%%%%%%%%
%%%   Motivation
%%%%%%%%%%%%%%%%%%%%%%%%%%%%%%%%%%%%%%%%%%%%%%%%%%%%%%%%%%%%%%%%%%%%%%%%%%%%%%%%%%%%%%%%%%%%%%%%%%
\subsection{Motivation: specific linear equations}
For $0<\lambda<1$,
let us now consider the following linear system of equations
\begin{equation}  \label{eq:specific_form}
{\hat{A}}_\lambda\vec{x}=\vec{b},
\end{equation}
where
\begin{equation} \label{eq:specific_form2}
{\hat{A}}_\lambda=
\left(
\begin{array}{cc}
\frac{1}{2} & \lambda-\frac{1}{2} \\
\lambda-\frac{1}{2} & \frac{1}{2} \\
\end{array}
\right),
\quad
\vec{b}=
\left(
\begin{array}{c}
1 \\ 0 \\ 
\end{array}
\right)=\ket{0}.
\end{equation}
Then one can readily check that ${\hat{A}}_\lambda$ is Hermitian,
and can obtain the solution of the equation ${\hat{A}}_\lambda\vec{x}=\vec{b}$
which is given by
\begin{equation} \label{eq:theoretical_sol}
\vec{x}
%=\frac{1}{2\lambda(1-\lambda)}(\ket{0}+(1-2\lambda)\ket{1})
=\frac{1}{\sqrt{2}\lambda}\ket{+}+\frac{1}{\sqrt{2}(1-\lambda)}\ket{-},
%\frac{1}{\sqrt{2}\lambda}\ket{+}+\frac{1}{\sqrt{2}(1-\lambda)}\ket{-},
\end{equation}
where %$N=2\lambda(1-\lambda)$ and 
$\ket{\pm}=(\ket{0}\pm\ket{1})/\sqrt{2}$.
%$\ket{\pm}=(\ket{0}\pm\ket{1})/\sqrt{2}$.

From the original HHL algorithm with ${\hat{A}}_\lambda$ and $\vec{b}$ in Eq.~(\ref{eq:specific_form2}), we can obtain the fidelity~\cite{W13} between the results of the algorithm with a $k$-qubit register
($k=1,2,3$) and analytical results given by 
\begin{equation}
F_k(\lambda)\equiv F(\rho_{k,\lambda},\psi_\lambda),
\end{equation}
where $F(\cdot,\cdot)$ indicates the quantum fidelity,
the final state $\rho_{k,\lambda}$ is the solution state describing the qubit system $V$
obtained at the end of algorithm performance,
and $\psi_\lambda$ is the normalized solution of $\vec{x}$ in Eq.~(\ref{eq:theoretical_sol}).
In Fig.~\ref{fig:graph_fidelity} (a),
the fidelities $F_k(\lambda)$ are presented with $k=1,2,3$
(details in Appendix~\ref{app:Fidelities}).
It indicates that more register qubits make a better and larger window
for higher fidelity between outcome states and the analytical solutions. 
From the curves,
we gain two features on the performance of the original HHL algorithm
which have not appeared in any previous literature.
The first feature is that
we can find an exact solution of the linear system of equations
only when the matrix ${\hat{A}}_\lambda$ is perfectly $n$-estimated.
In particular,
note that the fidelities reach to 1 with both 2- and 3-qubit registers
for $\lambda = 1/4,\,1/2,\, 3/4$.
In other words,
additional register qubits can increase the fidelity
if ${\hat{A}}_\lambda$ is not perfectly $n$-estimated.
The second one is that in the HHL algorithm
the use of a smaller size of register provides more precise solutions
at neighborhoods of the perfectly $n$-estimated eigenvalues.
For example,
we can verify that $F_3 <F_2 <F_1$ for $\lambda = 0.475$ in Fig.~\ref{fig:graph_fidelity} (a).

From the fact that $F_2(0.5)=F_3(0.5)=1$,
one may think that, for some restriction of $\lambda$,
circuit implementations for the HHL algorithm with 3-qubit register can be simplified
by using 2 qubits as register of the algorithm.
For example, the algorithm may be implemented by using a smaller number of gates,
and could have more efficient performance with the reduced amount of errors.
The idea motivates us to devise a quantum linear equation algorithm
whose circuit implementation is more simplified than that of the original HHL algorithm.

%%%%%%%%%%%%%%%%%%%%%%%%%%%%%%%%%%%%%%%%%%%%%%%%%%%%%%%%%%%%%%%%%%%%%%%%%%%%%%%%%%%%%%%%%%%%%%%%%%
%%%   Description of hybrid HHL algorithm
%%%%%%%%%%%%%%%%%%%%%%%%%%%%%%%%%%%%%%%%%%%%%%%%%%%%%%%%%%%%%%%%%%%%%%%%%%%%%%%%%%%%%%%%%%%%%%%%%%
\subsection{Description of hybrid HHL algorithm}

We here present the hybrid HHL algorithm, which mainly consists of the blocks of the quantum phase estimation algorithm (QPEA), classical computing, and a reduced HHL algorithm to test the original and hybrid HHL algorithms with a two-qubit register as described in Fig.~\ref{fig:circuit_two_algorithm}. In particular, the third part of the hybrid algorithm is called the \emph{reduced HHL part} because the part is not an independent quantum algorithm.

\begin{figure}
\centering
\includegraphics[width=.6\linewidth,trim=0cm 0cm 0cm 0cm]{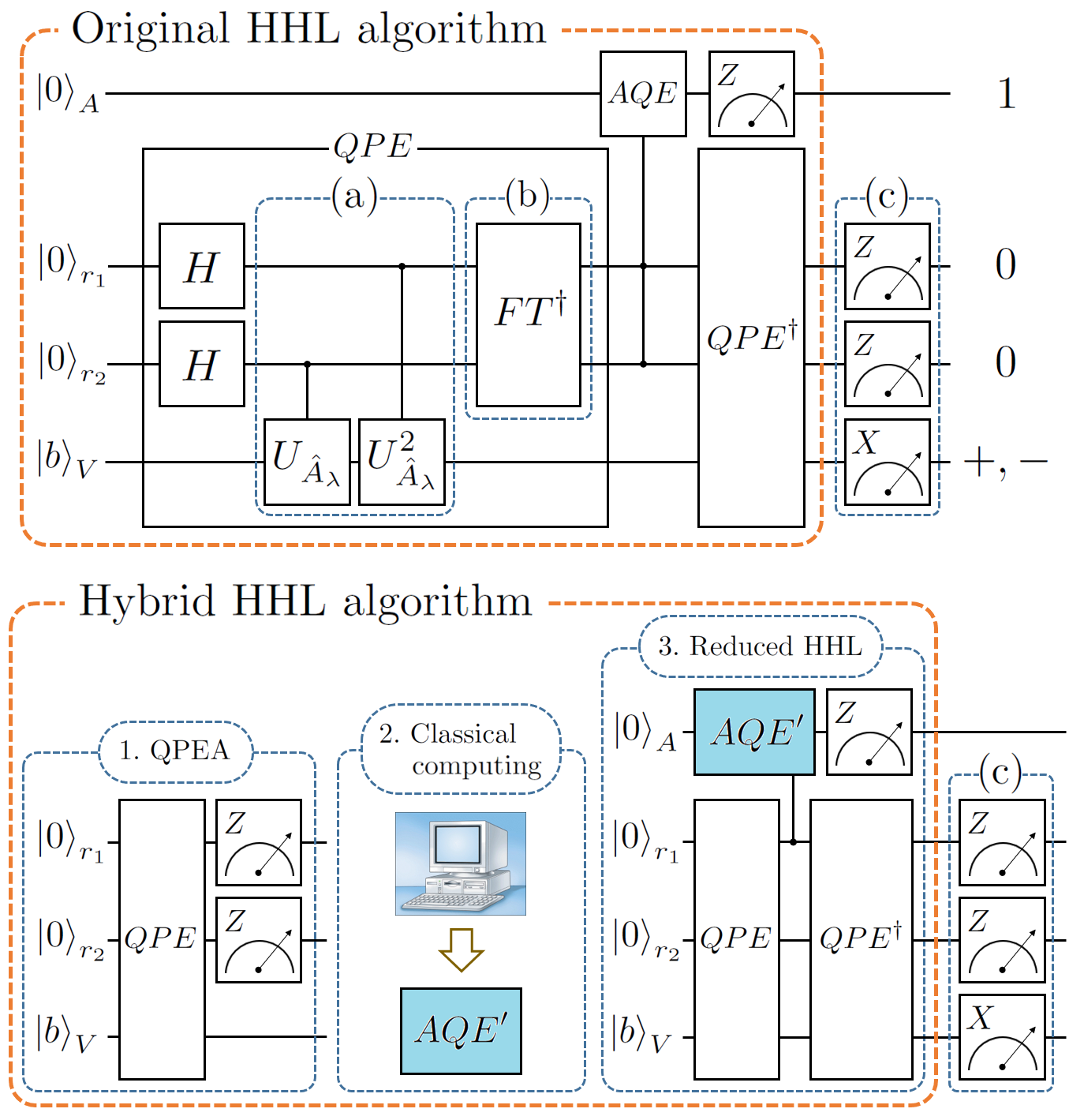}
\caption{
The circuit diagrams of the original and hybrid HHL algorithms
for ${\hat{A}}_\lambda$ and $\vec{b}$ in Eq.~(\ref{eq:specific_form2}).
(a) The controlled $U_\lambda^m$ gate,
where $U_{{\hat{A}_\lambda}}^m=(e^{2\pi i {\hat{A}}_\lambda})^m$ for $m\in\mathbb{Z}$.
(b) The inverse quantum Fourier transform for two qubits.
(c) Additional measurement devices to check outputs of the algorithms.
In the hybrid HHL algorithm,
$AQE'$ indicates a reduced AQE part (aqua color).
The detailed circuit implementations of (a), (b), and $AQE'$ are given in Fig.~\ref{fig:circuit_implementation}.
}
\label{fig:circuit_two_algorithm}
\end{figure}

\begin{itemize}
\item \textbf{QPEA:}
Repeatedly perform the QPEA to obtain $k$-bit classical information of eigenvalues with ${\hat{A}}_\lambda$ and $\ket{b}$.

\item \textbf{Classical Computing:}
Analyze measurement outcomes from the first step
by means of classical computers.
Based on the analyzed data,
such as an estimation of the probability distribution in Fig.~\ref{fig:graph_fidelity} (b),
one determines which simpler circuit implementation of the original AQE part,
called the \emph{reduced AQE part},
is applicable.
The circuit of the reduced AQE part is implemented by the classical analysis.

\item \textbf{Reduced HHL:} 
Perform 
the HHL algorithm with the reduced AQE part instead of the original AQE part.
\end{itemize}

Importantly, if the reduced AQE part is not applicable from the second step of the hybrid algorithm due to the lack of capability to distinguish different eigenvalues, the user of the algorithm should restart the first step with more register qubits to perform the reduced HHL part.

%%%%%%%%%%%%%%%%%%%%%%%%%%%%%%%%%%%%%%%%%%%%%%%%%%%%%%%%%%%%%%%%%%%%%%%%%%%%%%%%%%%%%%%%%%%%%%%%%%
%%%   How does hybrid HHL algorithm work?
%%%%%%%%%%%%%%%%%%%%%%%%%%%%%%%%%%%%%%%%%%%%%%%%%%%%%%%%%%%%%%%%%%%%%%%%%%%%%%%%%%%%%%%%%%%%%%%%%%
\subsection{How does hybrid HHL algorithm work?}
Let us consider that our hybrid HHL algorithm is applied
to the linear equation ${\hat{A}}_\lambda\vec{x}=\vec{b}$ in Eq.~(\ref{eq:specific_form})
when $\lambda=1/4,2/4,3/4$,
and that we use only 2 qubits as register of the hybrid HHL algorithm.
Assume that $\lambda=1/4,2/4,3/4$ is unknown.

In the first step of the hybrid HHL algorithm,
the QPEA with 2-qubit register is repeatedly performed
on the matrix ${\hat{A}}_\lambda$ and the state $\vec{b}$.
Then as depicted in Fig.~\ref{fig:graph_fidelity} (b),
we may obtain a probability distribution for the measurement outcomes of the QPEA given by
\begin{eqnarray}
\Pr(j0)
&=&
\frac{1}{16}e^{-6\pi i\lambda} ( 1+e^{4\pi i\lambda})^2 \left( (-1)^j + e^{2\pi i\lambda} \right)^2, 
\label{Prob1}\\
\Pr(01)
&=&
\Pr(11)
=-\frac{1}{8}e^{-4\pi i\lambda}\left( -1+e^{4\pi i\lambda}\right)^2,
\label{Prob2}
\end{eqnarray}
where $j=0,1$.

In the second step of the algorithm,
we can know what the eigenvalues of ${\hat{A}}_\lambda$ are
from the probability distribution in Fig.~\ref{fig:graph_fidelity} (b).
In addition,
we can also know that ${\hat{A}}_\lambda$ is perfectly $2$-estimated,
$\bar{m}_2$ $(=1)$ is a fixed eigenmean
for the matrices ${\hat{A}}_{1/4}$ and ${\hat{A}}_{3/4}$,
and the matrix ${\hat{A}}_{2/4}$ has fixed eigenmeans
$\bar{m}_1$ $(=1)$ and $\bar{m}_2$ $(=0)$.
From the classical information,
the AQE part of the HHL algorithm can be simply implemented.
In detail,
the AQE parts for the matrices ${\hat{A}}_{1/4}$ and ${\hat{A}}_{3/4}$
are given by a controlled-unitary operation
\begin{eqnarray} \label{eq:AQE_mapping_1_over_4}
\ket{0}_A\ket{0}_{r_1} &\mapsto&
\left(\sqrt{1-c^2}\ket{0}_A+c\ket{1}_A\right)\ket{0}_{r_1} \nonumber \\
\ket{0}_A\ket{1}_{r_1} &\mapsto&
\left(\sqrt{1-\left(\frac{c}{3}\right)^2}\ket{0}_A +\frac{c}{3}\ket{1}_A\right)\ket{1}_{r_1},
\end{eqnarray}
and the AQE for the matrix ${\hat{A}}_{2/4}$ is given by a single-qubit unitary operation
\begin{equation} \label{eq:AQE_mapping_2_over_4}
\ket{0}_A \mapsto\sqrt{1-\left(\frac{c}{2}\right)^2}\ket{0}_A+\frac{c}{2}\ket{1}_A.
\end{equation}

One of the practical drawbacks in the HHL algorithm is that we anyway need to know some partial information of the matrix $\hat{A}$ to setup the value $c$ in Eq.~(\ref{eq:n_reg_ancillary_quantum_encoding_mapping})
in the physical circuit of the AQE. Our main purpose is to extract this information with QPEA, and then the approximated value $c$ is now known for AQE and reduced AQE.
To implement the original and reduced AQE parts, we employ a specific conditional phase gate $R_y(\theta_n)$, where $R_y(\cdot)$ is the rotation gate about the $\hat{y}$ axis
and $\theta_n:=2\arccos(\sqrt{1-{{c_\lambda}^2}/{n^2}})$ 
with ${c_\lambda} = \parallel {{\hat{A}}_\lambda}^{-1} \ket{b} \parallel^{-1}$ for $n\ge1$.

In the third step,
by performing the reduced HHL part on the linear equation ${\hat{A}}_\lambda\vec{x}=\vec{b}$,
whose reduced AQE part is reconstructed
based on Eqs.~(\ref{eq:AQE_mapping_1_over_4}) and~(\ref{eq:AQE_mapping_2_over_4}),
we can obtain the normalized solution of the linear equation in the qubit system $V$.

In these examples,
our hybrid algorithm solves the linear equation under the condition that 
the matrix $\hat{A}_\lambda$ is perfectly $2$-estimated and it has fixed eigenmeans.
In fact,
this condition is indispensable for reducing the AQE part of the original HHL algorithm.
More generally,
the following theorem shows that
if a matrix $\hat{A}$ in Eq.~(\ref{eq:linear_system_of_equations}) is perfectly $n$-estimated,
and it has fixed eigenmeans,
then we can implement the AQE part by using smaller size of register
when the eigenvalues are known as follows.

\begin{Thm} \label{thm:reduced_AQE_part}
Let $n,k\in\mathbb{N}$ with $k\le n$.
If a matrix $\hat{A}$ is perfectly $n$-estimated,
and the matrix $\hat{A}$ has $k$ fixed eigenmeans,
then the AQE part can be implemented by $(n-k)$-qubit register.
\end{Thm}

\begin{proof}
Since each eigenvalue $\lambda_j$ of $\hat{A}$ is perfectly $n$-estimated,
its binary representation can be expressed as
\begin{equation}
\lambda_j=0.b_{1}^{j}b_{2}^{j}\cdots b_{n}^{j}
\end{equation}
for some $b_{i}^{j}\in\{0,1\}$.
Then since $\beta_{x|j}$ in Eq.~(\ref{eq:AQE_state}) becomes
\begin{equation}
\beta_{x|j}=\frac{1}{2^n}\sum_{y=0}^{2^n-1} e^{2\pi i y\left(2^n\lambda_j-{x}\right)/{2^n}}=1
\end{equation}
if $x=2^n\lambda_j$, and $\beta_{x|j}=0$ otherwise,
the state in Eq.~(\ref{eq:after_PE_part}) must be
\begin{equation} \label{eq:state_condition1}
\ket{0}_A\otimes\sum_{j=1}^{l}\alpha_j\ket{b_{1}^{j}b_{2}^{j}\cdots b_{n}^{j}}_R\otimes\ket{u_j}_V.
\end{equation}

Since the positions of the fixed eigenmeans of $\hat{A}$ do not affect this process,
without loss of generality, we may assume that
the $k$ fixed eigenmeans of the matrix $\hat{A}$
are $\bar{m}_1,\cdots,\bar{m}_k$.
Then the state in Eq.~(\ref{eq:state_condition1}) becomes
\begin{equation} \label{eq:state_condition2_k_less_n}
\ket{0}_A
\otimes\ket{\bar{m}_1\cdots\bar{m}_k}_{r_1\cdots r_k}
\otimes\sum_{j=1}^{l}\alpha_j\ket{b_{k+1}^{j}\cdots b_{n}^{j}}_{r_{k+1}\cdots r_n}
\otimes\ket{u_j}_V,
\end{equation}
since $\bar{m}_i=b_i^j$ holds for all $j$ and $1\le i\le k$.
Thus the AQE part can be implemented
by using $(n-k)$-qubit register as follows:
\begin{equation} \label{eq:n_mius_k_reg_ancillary_quantum_encoding_mapping}
\ket{0}_A\ket{y}_{r_{k+1}\cdots r_n}
\mapsto
\left(\sqrt{1-\frac{c^2}{(y'+y)^2}}\ket{0}_A+\frac{c}{y'+y}\ket{1}_A\right)
\ket{y}_{r_{k+1}\cdots r_n},
\end{equation}
where $0\le y\le2^{n-k}-1$ and $y'=\sum_{i=1}^{k}2^{(n-i)}\bar{m}_i$.
\end{proof}

Remark that Theorem~\ref{thm:reduced_AQE_part}
is useful for our hybrid algorithm as follows.
First of all,
the eigenvalues of $\hat{A}$ can be perfectly $n$-estimated
when a sufficiently large number of qubits are used as register of the HHL algorithm.
Secondly,
since the HHL algorithm deals with positive semidefinite matrices
whose eigenvalues are between 0 and 1 in our case,
the matrix $\hat{A}$ can have at least a fixed eigenmean.
Thus, by Theorem~\ref{thm:reduced_AQE_part}, 
the AQE part can be implemented with the reduced number of qubit register, 
depending on the number of fixed eigenmeans. 
In our hybrid algorithm 
we can correctly guess the eigenvalues of $\hat{A}$ with high probability
by means of the QPEA in advance of the HHL algorithm, 
and hence we can in practice implement the reduced AQE part.

%%%%%%%%%%%%%%%%%%%%%%%%%%%%%%%%%%%%%%%%%%%%%%%%%%%%%%%%%%%%%%%%%%%%%%%%%%%%%%%%%%%%%%%%%%%%%%%%%%
%%%                                                                                                                           
%%%   Circuit implementation and experiment
%%%                                                                                                                           
%%%%%%%%%%%%%%%%%%%%%%%%%%%%%%%%%%%%%%%%%%%%%%%%%%%%%%%%%%%%%%%%%%%%%%%%%%%%%%%%%%%%%%%%%%%%%%%%%%
\section{Circuit implementation and experiment} \label{sec:CIaE}

%%%%%%%%%%%%%%%%%%%%%%%%%%%%%%%%%%%%%%%%%%%%%%%%%%%%%%%%%%%%%%%%%%%%%%%%%%%%%%%%%%%%%%%%%%%%%%%%%%
%%%   IBM Quantum Experience
%%%%%%%%%%%%%%%%%%%%%%%%%%%%%%%%%%%%%%%%%%%%%%%%%%%%%%%%%%%%%%%%%%%%%%%%%%%%%%%%%%%%%%%%%%%%%%%%%%
\subsection{IBM Quantum Experience}
The IBMQX is the name of online facilities for general public
who can test their own experimental protocols in five (or sixteen) superconducting qubits.
Although its physical setup consists of a complex architecture built
by superconducting qubits and readout resonators in a single chip,
the user interface is designed with simple diagrams,
which represent single- and two-qubit gates, and 
is easy to understand and to write the programs
without much prior knowledge of quantum information theory and experimental setups.

We in particular use four qubits in the five-qubit systems (called IBMQX2 and IBMQX4)
and they have a different topology of connectivity for two-qubit gates.
For example,
they provide a controlled-NOT (CNOT) gate at the end-user level
but the physical two-qubit gate is actually performed by a cross-resonance gate \cite{CRgate,CRgate2}, which implies that additional single-qubit gates are required to match the desired CNOT gate.
Fortunately,
single-qubit gates in their transmon qubits are very accurate
and the fidelity of gate operations mostly depends on that of the cross-resonance gate
and the readout errors after the total quantum operation. For example, we utilize single-qubit $R_z$ gates for the algorithms as much as we can because this can be realized without applying any microwave but with only shifting the phase of the next applied microwave \cite{SarahPRA}. 

Because the IBMQX setup shows the daily small fluctuation of parameters,
they provide average device calibrations,
which might be useful for understanding the imperfection of the experimental data.
For example,
the transmon\rq{}s energy frequency (between $\ket{0}$ and $\ket{1}$) is roughly about 5 GHz,
which is fit to the microwave frequency with 6 cm wavelength.
Importantly,
one of the important measures for coherence time is $T_1 \approx 50\, \mu$s,
and it approximately limits the total operation time $t$ in performance of quantum processing
such as $\ket{1}\bra{1} \rightarrow e^{-t/T_1}\ket{1}\bra{1}$ for a single-qubit decay rate.
For example,
the CNOT gate (consisting of a cross-resonance gate and a few single-qubit gates) takes
around 200 ns,
and it roughly indicates that 50 times of CNOT gates might not exceed the fidelity 0.82
because $e^{-1/5} \approx 0.82$ at the current IBMQX setup.
Therefore,
the hybrid quantum algorithm might be beneficial
for experimental demonstrations under practical circumstances
because it has simpler quantum gates with the support of classical information processing.

\begin{figure}[t]
\centering
\includegraphics[width=.6\linewidth,trim=0cm 0cm 0cm 0cm]{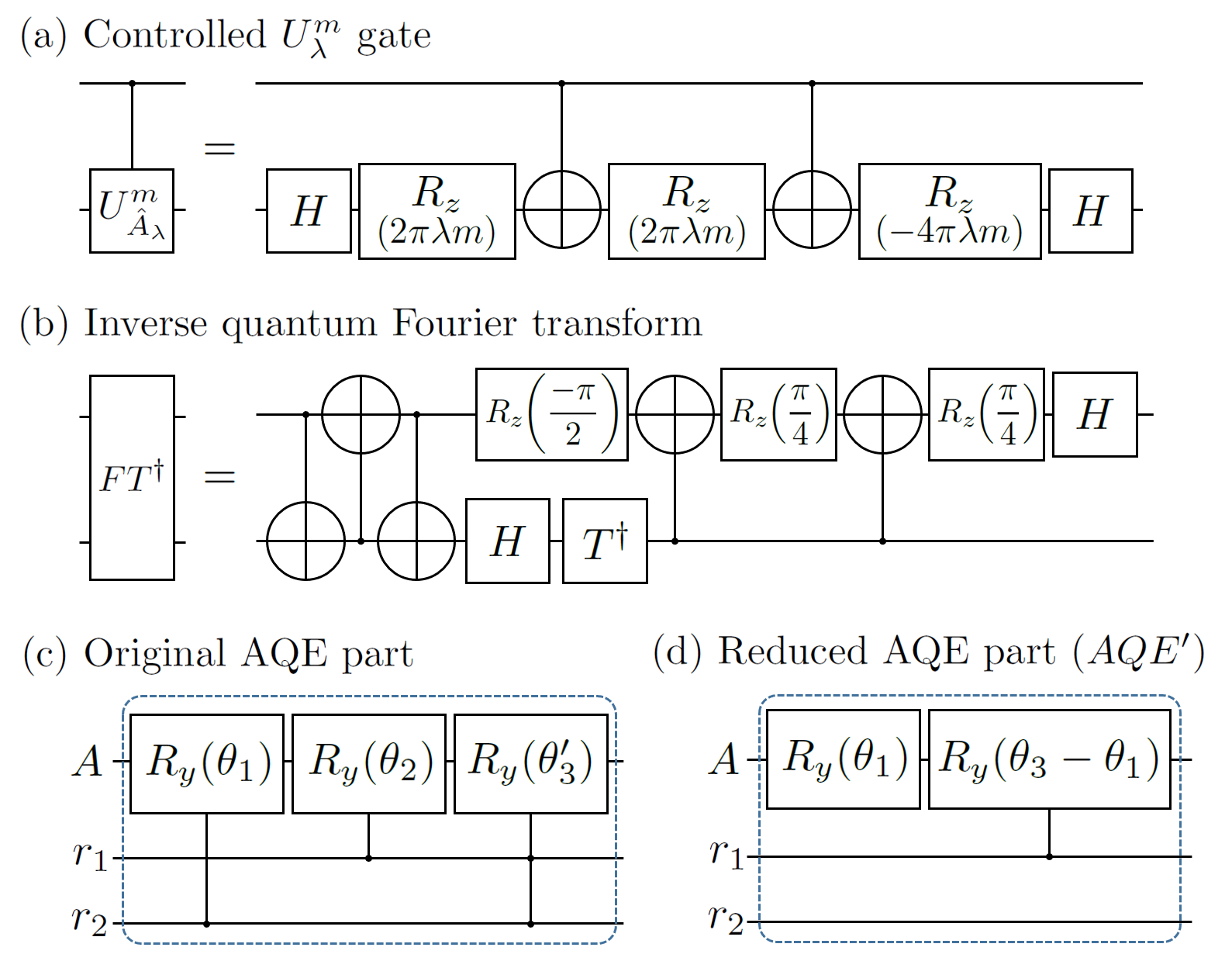}
\caption{
The circuit implementations on IBMQX setups:
(a) the controlled $U_\lambda^m$ gate
where $U_{{\hat{A}_\lambda}}^m=(e^{2\pi i {\hat{A}}_\lambda})^m$ for $m\in\mathbb{Z}$,
(b) the inverse quantum Fourier transform for two qubits,
(c) the original AQE part,
and (d) the reduced AQE part for the reduced HHL part when $\lambda=1/4,2/4$.
Here, $\theta'_3:=\theta_3-(\theta_1+\theta_2)$.
}
\label{fig:circuit_implementation}
\end{figure}
%%%%%%%%%%%%%%%%%%%%%%%%%%%%%%%%%%%%%%%%%%%%%%%%%%%%%%%%%%%%%%%%%%%%%%%%%%%%%%%%%%%%%%%%%%%%%%%%%%
%%%   Setups for circuit implementations on IBMQX
%%%%%%%%%%%%%%%%%%%%%%%%%%%%%%%%%%%%%%%%%%%%%%%%%%%%%%%%%%%%%%%%%%%%%%%%%%%%%%%%%%%%%%%%%%%%%%%%%%
\subsection{Setups for circuit implementations on IBMQX}
We now describe experimental setups of the hybird HHL algorithm
compared with the original HHL algorithm with two qubit register
to solve the linear equation given
by the parameterized matrix ${\hat{A}}_\lambda$ in Eq.~(\ref{eq:specific_form2}).
In addition,
we only deal with the matrices ${\hat{A}}_{1/4}$ and ${\hat{A}}_{2/4}$,
since ${\hat{A}}_{1/4}$ and ${\hat{A}}_{3/4}$ have the same eigenvalues.
In the IBMQX setups, it is also possible to test the algorithms by using a three-qubit register.
However, the complex circuit implementations dramatically decrease
the fidelities of the solutions beyond the analysis scope.
More importantly,
because the original and hybrid HHL algorithms exactly find
the same solution of ${\hat{A}}_\lambda\vec{x}=\vec{b}$
for the ideal (no-error) cases that $\lambda=1/4, 2/4$,
it is crucial to compare the performance of the original and hybrid HHL algorithms
under the IBMQX setups under error-propagating circumstances. Note that a similar experimental investigation has been recently shown
with fixed matrix $\hat{A}$, which cannot cover the class of our parameterized matrix in Eq.~(\ref{eq:specific_form2}) \cite{Other4QHHL}. 

As explained in Section \ref{sec:HHL} and Fig.~\ref{fig:circuit_two_algorithm},
the original HHL algorithm consists of the QPE, the AQE,
and the inverse QPE with a qubit measurement on the ancillary qubit,
as shown in the top of Fig.~\ref{fig:circuit_two_algorithm}.
The QPE part is mainly decomposed by the parts (a) and (b)
in Fig.~\ref{fig:circuit_two_algorithm} or Fig.~\ref{fig:circuit_implementation}.
The first part (a) consists of two controlled unitary gates
whose circuit implementations~\cite{BBC95} are found in Fig.~\ref{fig:circuit_implementation} (a).
The second part (b) is the inverse of the two-qubit QFT,
which is a combination of a SWAP gate, two CNOT gates, and some single-qubit gates 
shown in Fig.~\ref{fig:circuit_implementation} (b).
After the inverse QPE part,
if the ancillary qubit is measured in $\ket{1}_A$,
the register qubits always become $\ket{00}_{r_1 r_2}$ in principle.
However,
the propagated errors during the whole operation time might cause the other outcomes ($\neq \ket{00}_{r_1 r_2}$) in real experiments. This can be verified by setting the measurements of register qubits in Fig.~\ref{fig:circuit_two_algorithm} (c) to post-select successful outcomes.  

For the hybrid HHL algorithm in Fig.~\ref{fig:circuit_two_algorithm},
classical computing is sandwiched between two quantum computing parts.
The first part of the quantum algorithm is called QPEA
similar to the QPE part in the original HHL.
After the measurement of the two-qubit register (step 1), 
the analysis from classical computing decides the operation angles ($\theta_j$) in the reduced AQE circuit ($AQE\rq{}$) with respect to the measured first two digits in $r_1$ and $r_2$ (step 2). Finally, the chosen angles from the classical imformation are applied in the lightened circuit of $AQE \rq$ (step 3). The original and reduced AQE circuits are shown in Fig.~\ref{fig:circuit_implementation} (c) and (d), respectively.

Therefore,
we will show the experimental results of QPEA and the reduced HHL parts with $\lambda=1/4,1/2$
compared with the original HHL algorithm in the next subsection.
If we consider a general case that $\lambda\neq  k/4$ with $k=1,2,3$,
we cannot exactly estimate the eigenvalues of ${\hat{A}}_\lambda$
from the probability distribution given in Fig.~\ref{fig:graph_fidelity} (b) and request more register qubits for the algorithm, however,
it also indicates that a small variance of eigenvalues
($|\lambda - k/4| = \delta$ with small $\delta$)
would give us a high fidelity of the solution state (even better than using a three-qubit register in principle) as shown in Fig.~\ref{fig:graph_fidelity} (a).

\begin{figure}
\begin{minipage}[b]{.49\linewidth}
\includegraphics[width=\linewidth,trim=0cm 0cm 0cm 0cm]{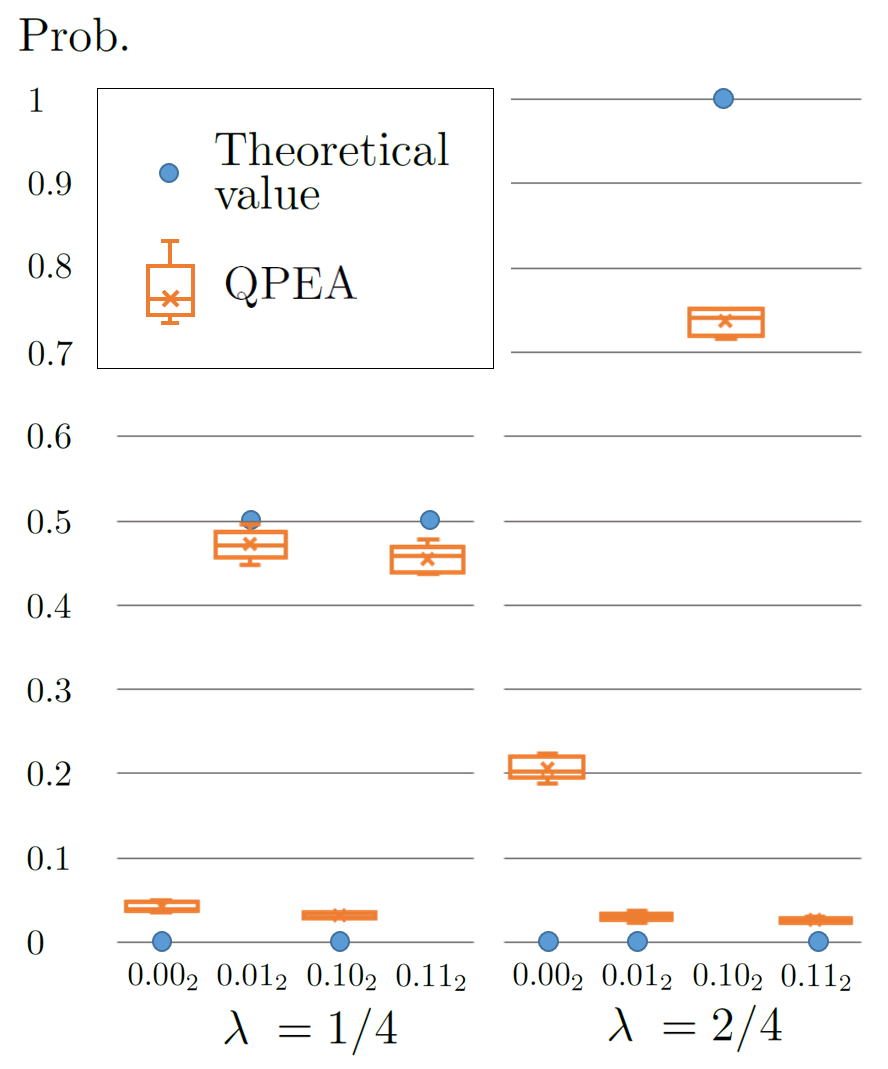}
\center{(a)}
\end{minipage}
\begin{minipage}[b]{.49\linewidth}
\includegraphics[width=\linewidth,trim=0cm 0cm 0cm 0cm]{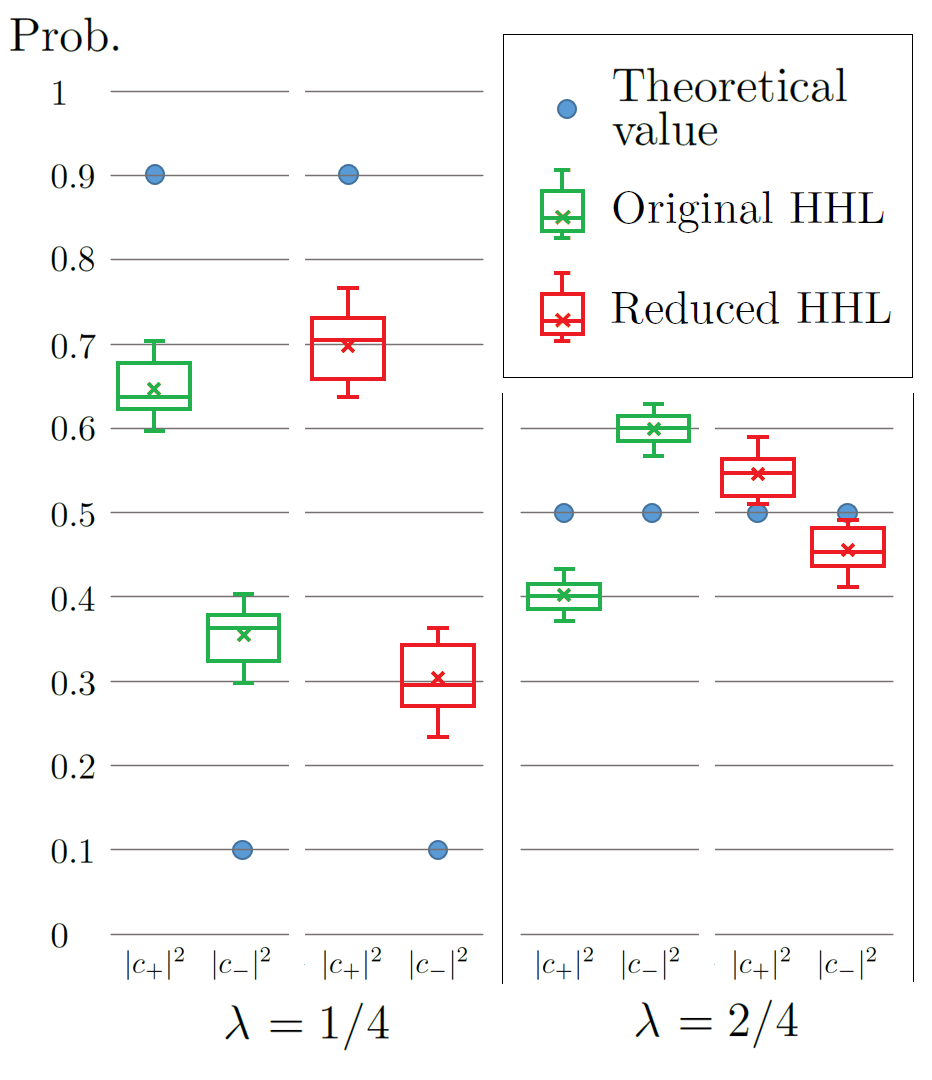}
\center{(b)}
\end{minipage}
\caption{
Experimental results on IBMQX4:
(a) the QPEA with 2-qubit register. %for $\lambda=1/4,2/4$ 
(b) the original HHL algorithm and reduced HHL part.
%for $\lambda=1/4,2/4$ 
}
\label{fig:Test_result_QPEA}
\end{figure}
%%%%%%%%%%%%%%%%%%%%%%%%%%%%%%%%%%%%%%%%%%%%%%%%%%%%%%%%%%%%%%%%%%%%%%%%%%%%%%%%%%%%%%%%%%%%%%%%%%
%%%   Experimental result: QPEA
%%%%%%%%%%%%%%%%%%%%%%%%%%%%%%%%%%%%%%%%%%%%%%%%%%%%%%%%%%%%%%%%%%%%%%%%%%%%%%%%%%%%%%%%%%%%%%%%%%
\subsection{Experimental results for QPEA and Reduced HHL parts}
We here examine the original and hybrid HHL algorithms at the setups of IBMQX4. 
The experiments of QPEA are performed by using six CNOTs,
and the original HHL algorithm requires 28 CNOT gates
while the reduced HHL algorithm now has 14 CNOT gates.
Thus,
this indicates the reduction of 14 CNOT gates from the original HHL algorithm.
Note that ten sets of experimental data are used for each $\lambda$
with 1024 single-shot readouts per set for individual algorithms.

The QPE in the original HHL and QPEA in the hybrid HHL commonly have (a) a set of controlled unitary operations with $U_{\lambda}^m=(e^{2\pi {\hat{A}}_\lambda})^m$
and (b) the inverse QFT for two qubits shown in Fig.~\ref{fig:circuit_implementation}. The only difference between them is the measurement part in QPEA.
The hybrid scheme first accepts the results of QPEA to estimate partial information of eigenvalues
in two bits used for building the $QAE\rq{}$ circuit
as shown in Fig.~\ref{fig:circuit_implementation} (d).
From the results depicted in Fig.~\ref{fig:Test_result_QPEA} (a),
we can verify that the performance of the QPEA on IBMQX is quite useful to confirm the first two bits of the eigenvalues of ${\hat{A}}_\lambda$ even with some unavoidable errors in the IBMQX circuit.

%%%%%%%%%%%%%%%%%%%%%%%%%%%%%%%%%%%%%%%%%%%%%%%%%%%%%%%%%%%%%%%%%%%%%%%%%%%%%%%%%%%%%%%%%%%%%%%%%%
%%%   Experimental result: Reduced HHL part
%%%%%%%%%%%%%%%%%%%%%%%%%%%%%%%%%%%%%%%%%%%%%%%%%%%%%%%%%%%%%%%%%%%%%%%%%%%%%%%%%%%%%%%%%%%%%%%%%%
%\subsection{Reduced HHL part}

In Fig.~\ref{fig:Test_result_QPEA} (b),
the probabilities of theoretical and two experimental cases are depicted
for both the original and reduced HHL algorithms.
The solution state $\ket{x}_V$ is measured in observable $X$
in the basis set of $\{\ket{+},\ket{-} \}$
to verify the experimental solution state for both algorithms.
We also perform ten sets of data with 1024 single shots per set.
Since $\ket{x}_V$ is represented by $c_{+}\ket{+}+c_{-}\ket{-}$,
where $c_{\pm}\in\mathbb{C}$ such that $|c_{+}|^2+|c_{-}|^2=1$,
probabilities $|c_{+}|^2$ and $|c_{-}|^2$ tell us
how close $\ket{x}_V$ is to the theoretical solution.
As mentioned earlier,
the solution of the linear equation is obtained
when the ancillary qubit is $\ket{1}$ after measured by $Z$.
Theoretically,
this probability of the post-measurement state $\ket{1}$ is quit small
in the case of the linear equation in Eq.~(\ref{eq:specific_form}). 
This means that only the minority data contribute to plot
the probabilities $|c_{+}|^2$ and $|c_{-}|^2$.

From Fig.~\ref{fig:Test_result_QPEA},
we can know that the solution of the reduced HHL part is more accurate than
that of the original HHL algorithm.
The figure shows that,
if we accept to use the first two bits of QPEA for the reduced AQE circuit,
we can conclude that the results of the hybrid algorithm are closer to the theoretical results
than that of the original HHL algorithm in the IBMQX setups.

%%%%%%%%%%%%%%%%%%%%%%%%%%%%%%%%%%%%%%%%%%%%%%%%%%%%%%%%%%%%%%%%%%%%%%%%%%%%%%%%%%%%%%%%%%%%%%%%%%
%%%                                                                                                                           
%%%   Conclusion
%%%                                                                                                                           
%%%%%%%%%%%%%%%%%%%%%%%%%%%%%%%%%%%%%%%%%%%%%%%%%%%%%%%%%%%%%%%%%%%%%%%%%%%%%%%%%%%%%%%%%%%%%%%%%%
\section{Conclusion} \label{sec:Conclusion}
% Summary
We have described the HHL algorithm which solves a quantized version of given linear equations.
We have especially analyzed the QPE part of the HHL algorithm,
and have devised the hybrid version of the HHL algorithm. Under the IBMQX setups,
we have shown that the hybrid algorithm can reduce the number of two-qubit gates, and thus has more enhanced performance than that of the HHL algorithm for some specific linear equations.

% Discussion 1
The hybrid HHL algorithm stems from the fact
that the QPE part of the HHL algorithm is identical to the QPEA without measurement.
It follows that the AQE part of the original HHL algorithm can be reconstructed
if we are able to obtain classical information from measurement outcomes of the prior QPEA
to solve a linear equation.
We remark that
an iterative QPEA~\cite{ZKRO13} can be used as the first step of our hybrid algorithm.
Since the iterative QPEA does not need the quantum Fourier transform for its implementation,
the small number of qubits is required.
So the use of the iterative QPEA would improve the resource efficiency of the hybrid algorithm. 
In addition,
there have been some results in literature~\cite{HBBWP07,DJSW07}
which generalize or improve the QPEA,
and we expect that combining these results with the hybrid algorithm
leads to other new hybrid quantum linear equation algorithms.
Finally,
there have been developed some quantum algorithms,
such as the quantum counting algorithm~\cite{BHT98},
the quantum machine learning algorithm~\cite{LGZ16},
and the high-order quantum algorithm~\cite{B14},
which have relevance to the QPEA or the HHL algorithm.
Hence, it would be interesting to find out hybrid versions for these algorithms.

Regarding quantum supremacy~\cite{P12},
IBM has currently announced a new term, \emph{quantum volume}~\cite{BBCGS17},
which measures the useful amount of quantum computing done
by a quantum device with specific number of qubits and error rate.
In addition,
error mitigation approaches \cite{error-mitigation,error-mitigation1,error-mitigation-exp}
have shown a new direction of managing the error accuracy for specific cases.
In order to apply this extrapolation scheme,
the amount of errors should be sufficiently small to claim
that the error-propagation curve is linear.
As mentioned earlier, 
the reduced HHL part of our hybrid algorithm can be implemented
by a smaller number of quantum gates,
which reduces the total error rate from the gates.
Hence,
we expect that the technique in the hybrid algorithm can be adopted in quantum algorithms
to show quantum supremacy.

%%%%%%%%%%%%%%%%%%%%%%%%%%%%%%%%%%%%%%%%%%%%%%%%%%%%%%%%%%%%%%%%%%%%%%%%%%%%%%%%%%%%%%%%%%%%%%%%%%
%%%                                                                                                                           
%%%   Acknowledgments
%%%                                                                                                                           
%%%%%%%%%%%%%%%%%%%%%%%%%%%%%%%%%%%%%%%%%%%%%%%%%%%%%%%%%%%%%%%%%%%%%%%%%%%%%%%%%%%%%%%%%%%%%%%%%%
\section{Acknowledgments}
This research was supported 
by Basic Science Research Program 
through the National Research Foundation of Korea (NRF) funded 
by the Ministry of Science and ICT (NRF-2016R1A2B4014928) and the MSIT (Ministry of Science and ICT), Korea, under the ITRC (Information Technology Research Center) support program (IITP-2018-2018-0-01402) supervised by the IITP (Institute for Information \& communications Technology Promotion).
JJ acknowledges support from the EPSRC National Quantum Technology Hub in Networked Quantum Information Technology
(EP/M013243/1).
We acknowledge use of the IBMQX for this work.
The views expressed are those of the authors 
and do not reflect the official policy or position of IBM 
or the IBMQX team.

%%%%%%%%%%%%%%%%%%%%%%%%%%%%%%%%%%%%%%%%%%%%%%%%%%%%%%%%%%%%%%%%%%%%%%%%%
%
%  Appendix
%
%%%%%%%%%%%%%%%%%%%%%%%%%%%%%%%%%%%%%%%%%%%%%%%%%%%%%%%%%%%%%%%%%%%%%%%%%
\appendix

%%%%%%%%%%%%%%%%%%%%%%%%%%%%%%%%%%%%%%%%%%%%%%%%%%%%%%%%%%%%%%%%%%%%%%%%%
%  Appendix: Quantum phase estimation algorithm
%%%%%%%%%%%%%%%%%%%%%%%%%%%%%%%%%%%%%%%%%%%%%%%%%%%%%%%%%%%%%%%%%%%%%%%%%
\section{Quantum phase estimation algorithm} \label{app:QPEA}
%As in Sec.~\ref{sec:prel},
Suppose that a matrix $\hat{A}$ is Hermitian
with an eigenvalue $\lambda$ in $(0,1)$ with respect to the corresponding eigenstate $\ket{u}$.
For the unitary operation $U_{\hat{A}}$ defined as in Eq.~(\ref{eq:unitary_from_A}),
we obtain
\begin{equation}
U_{\hat{A}}\ket{u}=e^{2\pi i \hat{A}}\ket{u}=e^{2\pi i \lambda}\ket{u}.
\end{equation}
The aim of the QPEA is to find out an estimated value of $\lambda$,
which is given by a binary string.
The QPEA is performed with the input eigenstate $\ket{u}$ and $n$-qubit register.
Then the estimated value of $\lambda$ is obtained by measuring this $n$-qubit register
as described in Fig.~\ref{fig:algorithm_QPEA}.

Specifically,
let us explain a process of the QPEA on $\hat{A}$ in Eq.~(\ref{eq:linear_system_of_equations})
and its eigenstate $\ket{u_j}$ in Eq.~(\ref{eq:spectral_decomp}).
The total input state of the QPEA is initialized
in a quantum state $\ket{0}_R^{\otimes n}\otimes\ket{u_j}_V$
and Hadamard operations are firstly performed in $n$-qubit register,
as shown in Fig.~\ref{fig:algorithm_QPEA}.
After $n$ controlled unitary gates, controlled-$U_{\hat{A}}^{2^{n-1}}$,
the state (a) in Fig.~\ref{fig:algorithm_QPEA} is given in 
\begin{equation}
\frac{1}{\sqrt{2^n}}\sum_{y=0}^{2^n-1} e^{2\pi i \lambda_j y}\ket{y}_R\otimes\ket{u_j}_V.
\end{equation}
Then,
the inverse of quantum Fourier transform is applied in the register qubits,
and the state (b) in Fig.~\ref{fig:algorithm_QPEA} can be written in
\begin{equation}
\frac{1}{2^n}\sum_{x,y=0}^{2^n-1}
e^{{2\pi i y\left(\lambda_j-\frac{x}{2^n}\right)}} \ket{x}_R\otimes\ket{u_j}_V.
\label{1st_QPEA}
\end{equation}
Finally, each qubit in the register system $R$ is measured with observable $Z$.
For large $n$,
if the measured outcome $x$ in register qubits is close to $\lambda_j(n)$,
we find that $\Pr(x)$ is also close to one.
Otherwise,
if $x$ is close to another $n$-bit string which is not $\lambda_j(n)$,
$\Pr(x)$ is close to zero. 
Therefore for sufficiently large $n$, 
we are able to obtain the $n$-binary estimation $\lambda_j(n)$ of $\lambda_j$
from the probability distribution of the measurement outcomes.

\begin{figure}
\centering
\includegraphics[width=.45\linewidth,trim=0cm 0cm 0cm 0cm]{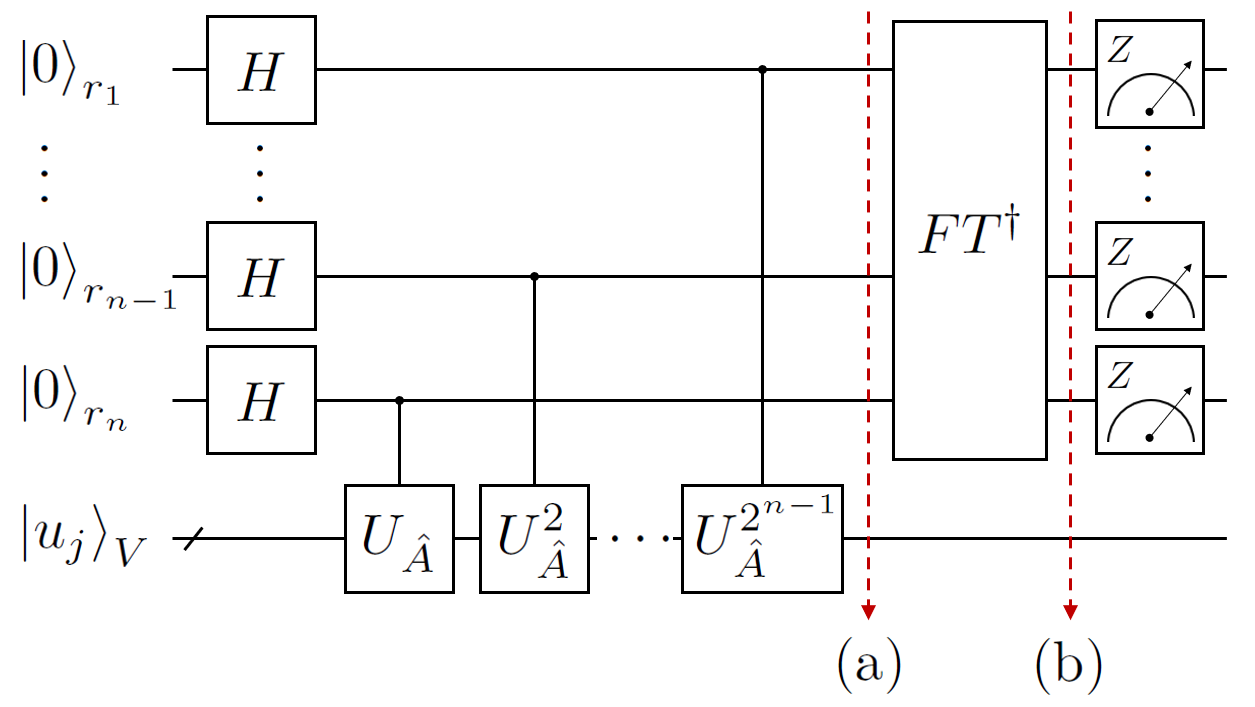}
\caption{
The circuit diagram of the QPEA with $n$-qubit register:
$H$ and $FT^\dagger$  are the Hadamard gate and the inverse quantum Fourier transform.
At the end, every register qubit is measured in observable $Z$.
}
\label{fig:algorithm_QPEA}
\end{figure}

%%%%%%%%%%%%%%%%%%%%%%%%%%%%%%%%%%%%%%%%%%%%%%%%%%%%%%%%%%%%%%%%%%%%%%%%%
%  Appendix: Fidelities in Fig.~\ref{fig:graph_fidelity}
%%%%%%%%%%%%%%%%%%%%%%%%%%%%%%%%%%%%%%%%%%%%%%%%%%%%%%%%%%%%%%%%%%%%%%%%%
\section{Fidelities in Fig.~\ref{fig:graph_fidelity} (a)} \label{app:Fidelities}
We here present explicit expressions of the fidelities in Fig.~\ref{fig:graph_fidelity} (a).
For $i=1,2,3$ and $\lambda\in(0,1)$,
denote $F_i(\lambda)$ as the fidelity
between an exact normalized solution and the solution obtained
from the HHL algorithm with $i$-qubit register.
Let $t_{\lambda}=e^{2 i \pi \lambda}$, and $t_{\lambda}^*$ be the complex conjugate of $t_{\lambda}$.
Then

\begin{equation}
F_1(\lambda)=
\frac{1}{2}
\left(
1 + (t_{\lambda}+t_{\lambda}^*)
\frac{ (-1 + \lambda) \lambda}{1 - 2 \lambda + 2 \lambda^2}
\right). \nonumber
\end{equation}

Let $X_{\lambda}=(40 + 32 i) - (129 + 64 i) \lambda + 129 \lambda^2$
and $Y_{\lambda}=(9 + 32 i) - (146 + 64 i) \lambda + 146 \lambda^2$.
Then
\begin{eqnarray}
F_2(\lambda)&=&
(t_{\lambda}^*)^3
[(
25 + 80 t_{\lambda} + 171 t_{\lambda}^2 + 171 t_{\lambda}^8 + 80 t_{\lambda}^9
+ 25 t_{\lambda}^{10}) (-1 + \lambda) \lambda +  4 t_{\lambda}^4 X_{\lambda}
+ 4 t_{\lambda}^6 X_{\lambda}^* + 2 t_{\lambda}^3 Y_{\lambda} + 2 t_{\lambda}^7 Y_{\lambda}^*
 \nonumber \\
& &+ 4 t_{\lambda}^5 (89 - 170 \lambda + 170 \lambda^2)]/
[4 (9 + 80 t_{\lambda} + 178 t_{\lambda}^2
+ 80 t_{\lambda}^3 + 9 t_{\lambda}^4) (1 - 2 \lambda + 2 \lambda^2)]. \nonumber
\end{eqnarray}

Let $A=140+105i$,
$B=(208+128i)\sqrt{2}$,
$C=8(35 + 52\sqrt{2})$,
$D=8(-35 + 52\sqrt{2})$,
$E=2(105 + 128\sqrt{2})$,
$F=-210 + 256\sqrt{2}$,
$G=11025 + 76672\sqrt{2}$,
and $H=-11025 + 76672\sqrt{2}$,
and let
$\alpha_{\lambda}=-(t_{\lambda}^{*})^{14}/(128(1 - 2 \lambda + 2 \lambda^2))$,
$\beta_{\lambda}=\alpha_{\lambda}(-1 + t_{\lambda}^{8})^2$,
$\gamma_{\lambda}=350 + 608i - 700\lambda$,
$\phi_{\lambda}=315 + 420i + (384 - 624i)\sqrt{2} - 3E\lambda$,
$\xi_{\lambda}=-304 + 175i + 608\lambda$.
Then we have the fidelity
\begin{equation}
F_3(\lambda)=\frac{\sum_{j=1}^{8} N_j(\lambda)}{D(\lambda)}, \nonumber
\end{equation}
where
\begin{eqnarray}
N_1(\lambda)&=&
\alpha_{\lambda}[
A^*+ B - 1276it_{\lambda}^{3}
+ 8712it_{\lambda}^{7} - 1276it_{\lambda}^{11}- C\lambda
+ (6t_{\lambda}^{5}+2t_{\lambda}^{13})\xi_{\lambda}^*
- (2t+6t_{\lambda}^{9})\xi_{\lambda} \nonumber \\
& & + (5t_{\lambda}^{4}+3t_{\lambda}^{12})
(A^* - B +D\lambda)- (3t_{\lambda}^{2}+5t_{\lambda}^{10})
(A -B^* +D\lambda) - 7t_{\lambda}^{8}(-A^* - B +C\lambda) \nonumber \\
& &+ (7t_{\lambda}^{6}+t_{\lambda}^{14})
(-A -B^* +C\lambda) 
]^2, \nonumber \\
N_2(\lambda)&=&
\beta_{\lambda}[
A^*i + Bi + F\lambda + t_{\lambda}^{5}\gamma_{\lambda}^*
+ t\gamma_{\lambda} -1276t_{\lambda}^{3}(-1 + 2 \lambda) + t_{\lambda}^{4}\phi_{\lambda}
+ t_{\lambda}^{2}{\phi_{\lambda}}^* + t_{\lambda}^{6}(-Ai -B^*i + F\lambda)
]^2, \nonumber \\
N_3(\lambda)&=&
\beta_{\lambda}[
A^*i +Bi + F\lambda
+ t_{\lambda}^{5}\gamma_{\lambda}^*
+ t_{\lambda}\gamma_{\lambda}
+ t_{\lambda}^{4}(A^*i -Bi -E\lambda)
+ t_{\lambda}^{2}(-Ai +B^*i -E\lambda)
+ t_{\lambda}^{6}(-Ai -B^*i + F\lambda)
]^2, \nonumber \\
N_4(\lambda)&=&
\beta_{\lambda}[
-A^* - B + C\lambda
+ 2t_{\lambda}^{5}\xi_{\lambda}^*
+ 2t_{\lambda}\xi_{\lambda}
+ t_{\lambda}^{4}(A^* - B +D\lambda)
+ t_{\lambda}^{2}(A -B^* +D\lambda)
+ t_{\lambda}^{6}(-A -B^* +C\lambda)
]^2, \nonumber \\
N_5(\lambda)&=&
\beta_{\lambda}[
A^*i +Bi + F\lambda
+ t_{\lambda}^{4}(A^*i -Bi -E\lambda)
+ t_{\lambda}^{2}(-Ai +B^*i -E\lambda)
+ t_{\lambda}^{6}(-Ai -B^*i + F\lambda)
]^2, \nonumber \\
N_6(\lambda)&=&
\beta_{\lambda}[
-A^* - B + C\lambda
+ t_{\lambda}^{4}(A^* - B +D\lambda)
+ t_{\lambda}^{2}(A -B^* +D\lambda)
+ t_{\lambda}^{6}(-A -B^* +C\lambda)
]^2, \nonumber \\
N_7(\lambda)&=&
\beta_{\lambda}[
-A^* - B + C\lambda
+ t_{\lambda}^{4}(-A^* + B -D\lambda)
- t_{\lambda}^{2}(A -B^* +D\lambda)
+ t_{\lambda}^{6}(-A -B^* +C\lambda)
]^2, \nonumber \\
N_8(\lambda)&=&
\beta_{\lambda}[
A^*i +Bi + F\lambda
+ t_{\lambda}^{6}(-Ai -B^*i + F\lambda)
+ t_{\lambda}^{2}(A^*i -B^*i + 2E\lambda)
+ t_{\lambda}^{4}(-A^*i +Bi + 2E\lambda)
]^2, \nonumber \\
D(\lambda)&=&
(t_{\lambda}^*)^{7}[
H - 75950t_{\lambda} -3Gt_{\lambda}^{2}
- 586524t_{\lambda}^{3} -5Gt_{\lambda}^{4}
- 227850t_{\lambda}^{5} +7Ht_{\lambda}^{6} \nonumber \\
& &+2133448t_{\lambda}^{7} +7Ht_{\lambda}^{8}
- 227850t_{\lambda}^{9} -5Gt_{\lambda}^{10}
- 586524t_{\lambda}^{11} -3Gt_{\lambda}^{12}
- 75950t_{\lambda}^{13} + Ht_{\lambda}^{14}
]. \nonumber
\end{eqnarray}

\bibliography{IBMQX}

\begin{thebibliography}{23}
\expandafter\ifx\csname natexlab\endcsname\relax\def\natexlab#1{#1}\fi
\expandafter\ifx\csname bibnamefont\endcsname\relax
  \def\bibnamefont#1{#1}\fi
\expandafter\ifx\csname bibfnamefont\endcsname\relax
  \def\bibfnamefont#1{#1}\fi
\expandafter\ifx\csname citenamefont\endcsname\relax
  \def\citenamefont#1{#1}\fi
\expandafter\ifx\csname url\endcsname\relax
  \def\url#1{\texttt{#1}}\fi
\expandafter\ifx\csname urlprefix\endcsname\relax\def\urlprefix{URL }\fi
\providecommand{\bibinfo}[2]{#2}
\providecommand{\eprint}[2][]{\url{#2}}

\bibitem[{\citenamefont{Shor}(1997)}]{S97}
\bibinfo{author}{\bibfnamefont{P.~W.} \bibnamefont{Shor}},
  \bibinfo{journal}{SIAM J. Comput.} \textbf{\bibinfo{volume}{26}},
  \bibinfo{pages}{1484–1509} (\bibinfo{year}{1997}).

\bibitem[{\citenamefont{Harrow et~al.}(2009)\citenamefont{Harrow, Hassidim, and
  Lloyd}}]{HHL09}
\bibinfo{author}{\bibfnamefont{A.~W.} \bibnamefont{Harrow}},
  \bibinfo{author}{\bibfnamefont{A.}~\bibnamefont{Hassidim}}, \bibnamefont{and}
  \bibinfo{author}{\bibfnamefont{S.}~\bibnamefont{Lloyd}},
  \bibinfo{journal}{Phys. Rev. Lett.} \textbf{\bibinfo{volume}{103}},
  \bibinfo{pages}{150502} (\bibinfo{year}{2009}).

\bibitem[{\citenamefont{Lloyd et~al.}(2016)\citenamefont{Lloyd, Garnerone, and
  Zanardi}}]{LGZ16}
\bibinfo{author}{\bibfnamefont{S.}~\bibnamefont{Lloyd}},
  \bibinfo{author}{\bibfnamefont{S.}~\bibnamefont{Garnerone}},
  \bibnamefont{and} \bibinfo{author}{\bibfnamefont{P.}~\bibnamefont{Zanardi}},
  \bibinfo{journal}{Nat. Commun.} \textbf{\bibinfo{volume}{7}}
  (\bibinfo{year}{2016}).

\bibitem[{\citenamefont{Berry}(2014)}]{B14}
\bibinfo{author}{\bibfnamefont{D.~W.} \bibnamefont{Berry}},
  \bibinfo{journal}{J. Phys. A: Math. Theor.} \textbf{\bibinfo{volume}{47}},
  \bibinfo{pages}{105301} (\bibinfo{year}{2014}).

\bibitem[{\citenamefont{Wilde}(2013)}]{W13}
\bibinfo{author}{\bibfnamefont{M.~M.} \bibnamefont{Wilde}},
  \emph{\bibinfo{title}{Quantum Information Theory}}
  (\bibinfo{publisher}{Cambridge University Press}, \bibinfo{year}{2013}).

\bibitem[{\citenamefont{Kitaev}()}]{K95}
\bibinfo{author}{\bibfnamefont{A.~Y.} \bibnamefont{Kitaev}},
  \eprint{arXiv:quant-ph/9511026}.

\bibitem[{\citenamefont{Massar and Popescu}(1995)}]{MP95}
\bibinfo{author}{\bibfnamefont{S.}~\bibnamefont{Massar}} \bibnamefont{and}
  \bibinfo{author}{\bibfnamefont{S.}~\bibnamefont{Popescu}},
  \bibinfo{journal}{Phys. Rev. Lett.} \textbf{\bibinfo{volume}{74}},
  \bibinfo{pages}{1259} (\bibinfo{year}{1995}).

\bibitem[{\citenamefont{Derka et~al.}(1998)\citenamefont{Derka, Buz̆ek, and
  Ekert}}]{DBE98}
\bibinfo{author}{\bibfnamefont{R.}~\bibnamefont{Derka}},
  \bibinfo{author}{\bibfnamefont{V.}~\bibnamefont{Buz̆ek}}, \bibnamefont{and}
  \bibinfo{author}{\bibfnamefont{A.~K.} \bibnamefont{Ekert}},
  \bibinfo{journal}{Phys. Rev. Lett.} \textbf{\bibinfo{volume}{80}},
  \bibinfo{pages}{1571} (\bibinfo{year}{1998}).

\bibitem[{\citenamefont{Huelga et~al.}(1997)\citenamefont{Huelga, Macchiavello,
  Pellizzari, Ekert, Plenio, and Cirac}}]{HMPEPC97}
\bibinfo{author}{\bibfnamefont{S.~F.} \bibnamefont{Huelga}},
  \bibinfo{author}{\bibfnamefont{C.}~\bibnamefont{Macchiavello}},
  \bibinfo{author}{\bibfnamefont{T.}~\bibnamefont{Pellizzari}},
  \bibinfo{author}{\bibfnamefont{A.~K.} \bibnamefont{Ekert}},
  \bibinfo{author}{\bibfnamefont{M.~B.} \bibnamefont{Plenio}},
  \bibnamefont{and} \bibinfo{author}{\bibfnamefont{J.~I.} \bibnamefont{Cirac}},
  \bibinfo{journal}{Phys. Rev. Lett.} \textbf{\bibinfo{volume}{79}},
  \bibinfo{pages}{3865} (\bibinfo{year}{1997}).

\bibitem[{\citenamefont{Rigetti and Devoret}(2010)}]{CRgate}
\bibinfo{author}{\bibfnamefont{C.}~\bibnamefont{Rigetti}} \bibnamefont{and}
  \bibinfo{author}{\bibfnamefont{M.}~\bibnamefont{Devoret}},
  \bibinfo{journal}{Phys. Rev. B} \textbf{\bibinfo{volume}{81}},
  \bibinfo{pages}{134507} (\bibinfo{year}{2010}).

\bibitem[{\citenamefont{Chow et~al.}(2011)\citenamefont{Chow, Córcoles,
  Gambetta, Rigetti, Johnson, Smolin, Rozen, Keefe, Rothwell, Ketchen
  et~al.}}]{CRgate2}
\bibinfo{author}{\bibfnamefont{J.~M.} \bibnamefont{Chow}},
  \bibinfo{author}{\bibfnamefont{A.~D.} \bibnamefont{Córcoles}},
  \bibinfo{author}{\bibfnamefont{J.~M.} \bibnamefont{Gambetta}},
  \bibinfo{author}{\bibfnamefont{C.}~\bibnamefont{Rigetti}},
  \bibinfo{author}{\bibfnamefont{B.~R.} \bibnamefont{Johnson}},
  \bibinfo{author}{\bibfnamefont{J.~A.} \bibnamefont{Smolin}},
  \bibinfo{author}{\bibfnamefont{J.~R.} \bibnamefont{Rozen}},
  \bibinfo{author}{\bibfnamefont{G.~A.} \bibnamefont{Keefe}},
  \bibinfo{author}{\bibfnamefont{M.~B.} \bibnamefont{Rothwell}},
  \bibinfo{author}{\bibfnamefont{M.~B.} \bibnamefont{Ketchen}},
  \bibnamefont{et~al.}, \bibinfo{journal}{Phys. Rev. Lett.}
  \textbf{\bibinfo{volume}{107}}, \bibinfo{pages}{080502}
  (\bibinfo{year}{2011}).

\bibitem[{\citenamefont{McKay et~al.}(2017)\citenamefont{McKay, Wood, Sheldon,
  Chow, and Gambetta}}]{SarahPRA}
\bibinfo{author}{\bibfnamefont{D.~C.} \bibnamefont{McKay}},
  \bibinfo{author}{\bibfnamefont{C.~J.} \bibnamefont{Wood}},
  \bibinfo{author}{\bibfnamefont{S.}~\bibnamefont{Sheldon}},
  \bibinfo{author}{\bibfnamefont{J.~M.} \bibnamefont{Chow}}, \bibnamefont{and}
  \bibinfo{author}{\bibfnamefont{J.~M.} \bibnamefont{Gambetta}},
  \bibinfo{journal}{Phys. Rev. A} \textbf{\bibinfo{volume}{96}},
  \bibinfo{pages}{022330} (\bibinfo{year}{2017}).

\bibitem[{\citenamefont{Zheng et~al.}(2017)\citenamefont{Zheng, Song, Chen,
  Xia, Liu, Guo, Zhang, Xu, Deng, Huang et~al.}}]{Other4QHHL}
\bibinfo{author}{\bibfnamefont{Y.}~\bibnamefont{Zheng}},
  \bibinfo{author}{\bibfnamefont{C.}~\bibnamefont{Song}},
  \bibinfo{author}{\bibfnamefont{M.-C.} \bibnamefont{Chen}},
  \bibinfo{author}{\bibfnamefont{B.}~\bibnamefont{Xia}},
  \bibinfo{author}{\bibfnamefont{W.}~\bibnamefont{Liu}},
  \bibinfo{author}{\bibfnamefont{Q.}~\bibnamefont{Guo}},
  \bibinfo{author}{\bibfnamefont{L.}~\bibnamefont{Zhang}},
  \bibinfo{author}{\bibfnamefont{D.}~\bibnamefont{Xu}},
  \bibinfo{author}{\bibfnamefont{H.}~\bibnamefont{Deng}},
  \bibinfo{author}{\bibfnamefont{K.}~\bibnamefont{Huang}},
  \bibnamefont{et~al.}, \bibinfo{journal}{Phys. Rev. Lett.}
  \textbf{\bibinfo{volume}{118}}, \bibinfo{pages}{210504}
  (\bibinfo{year}{2017}).

\bibitem[{\citenamefont{Barenco et~al.}(1995)\citenamefont{Barenco, Bennett,
  Cleve, DiVincenzo, Margolus, Shor, Sleator, Smolin, and Weinfurter}}]{BBC95}
\bibinfo{author}{\bibfnamefont{A.}~\bibnamefont{Barenco}},
  \bibinfo{author}{\bibfnamefont{C.~H.} \bibnamefont{Bennett}},
  \bibinfo{author}{\bibfnamefont{R.}~\bibnamefont{Cleve}},
  \bibinfo{author}{\bibfnamefont{D.~P.} \bibnamefont{DiVincenzo}},
  \bibinfo{author}{\bibfnamefont{N.}~\bibnamefont{Margolus}},
  \bibinfo{author}{\bibfnamefont{P.}~\bibnamefont{Shor}},
  \bibinfo{author}{\bibfnamefont{T.}~\bibnamefont{Sleator}},
  \bibinfo{author}{\bibfnamefont{J.~A.} \bibnamefont{Smolin}},
  \bibnamefont{and}
  \bibinfo{author}{\bibfnamefont{H.}~\bibnamefont{Weinfurter}},
  \bibinfo{journal}{Phys. Rev. A} \textbf{\bibinfo{volume}{52}},
  \bibinfo{pages}{3457} (\bibinfo{year}{1995}).

\bibitem[{\citenamefont{Zhou et~al.}(2013)\citenamefont{Zhou, Kalasuwan, Ralph,
  and O'Brien}}]{ZKRO13}
\bibinfo{author}{\bibfnamefont{X.-Q.} \bibnamefont{Zhou}},
  \bibinfo{author}{\bibfnamefont{P.}~\bibnamefont{Kalasuwan}},
  \bibinfo{author}{\bibfnamefont{T.~C.} \bibnamefont{Ralph}}, \bibnamefont{and}
  \bibinfo{author}{\bibfnamefont{J.~L.} \bibnamefont{O'Brien}},
  \bibinfo{journal}{Nature Photonics} \textbf{\bibinfo{volume}{7}},
  \bibinfo{pages}{223–228} (\bibinfo{year}{2013}).

\bibitem[{\citenamefont{Higgins et~al.}(2007)\citenamefont{Higgins, Berry,
  Bartlett, Wiseman, and Pryde}}]{HBBWP07}
\bibinfo{author}{\bibfnamefont{B.~L.} \bibnamefont{Higgins}},
  \bibinfo{author}{\bibfnamefont{D.~W.} \bibnamefont{Berry}},
  \bibinfo{author}{\bibfnamefont{S.~D.} \bibnamefont{Bartlett}},
  \bibinfo{author}{\bibfnamefont{H.~M.} \bibnamefont{Wiseman}},
  \bibnamefont{and} \bibinfo{author}{\bibfnamefont{G.~J.} \bibnamefont{Pryde}},
  \bibinfo{journal}{Nature} \textbf{\bibinfo{volume}{450}},
  \bibinfo{pages}{393–396} (\bibinfo{year}{2007}).

\bibitem[{\citenamefont{Dob\v{s}\'{i}\v{c}ek
  et~al.}(2007)\citenamefont{Dob\v{s}\'{i}\v{c}ek, Johansson, Shumeiko, and
  Wendin}}]{DJSW07}
\bibinfo{author}{\bibfnamefont{M.}~\bibnamefont{Dob\v{s}\'{i}\v{c}ek}},
  \bibinfo{author}{\bibfnamefont{G.}~\bibnamefont{Johansson}},
  \bibinfo{author}{\bibfnamefont{V.}~\bibnamefont{Shumeiko}}, \bibnamefont{and}
  \bibinfo{author}{\bibfnamefont{G.}~\bibnamefont{Wendin}},
  \bibinfo{journal}{Phys. Rev. A} \textbf{\bibinfo{volume}{76}},
  \bibinfo{pages}{030306(R)} (\bibinfo{year}{2007}).

\bibitem[{\citenamefont{Brassard et~al.}(1998)\citenamefont{Brassard, HØyer,
  and Tapp}}]{BHT98}
\bibinfo{author}{\bibfnamefont{G.}~\bibnamefont{Brassard}},
  \bibinfo{author}{\bibfnamefont{P.}~\bibnamefont{HØyer}}, \bibnamefont{and}
  \bibinfo{author}{\bibfnamefont{A.}~\bibnamefont{Tapp}},
  \emph{\bibinfo{title}{Quantum counting}}, vol. \bibinfo{volume}{1443}
  (\bibinfo{publisher}{Springer, Berlin, Heidelberg}, \bibinfo{year}{1998}).

\bibitem[{\citenamefont{Preskill}()}]{P12}
\bibinfo{author}{\bibfnamefont{J.}~\bibnamefont{Preskill}},
  \eprint{arXiv:1203.5813}.

\bibitem[{\citenamefont{Bishop et~al.}(2017)\citenamefont{Bishop, Bravyi,
  Cross, Gambetta, and Smolin}}]{BBCGS17}
\bibinfo{author}{\bibfnamefont{L.~S.} \bibnamefont{Bishop}},
  \bibinfo{author}{\bibfnamefont{S.}~\bibnamefont{Bravyi}},
  \bibinfo{author}{\bibfnamefont{A.}~\bibnamefont{Cross}},
  \bibinfo{author}{\bibfnamefont{J.~M.} \bibnamefont{Gambetta}},
  \bibnamefont{and} \bibinfo{author}{\bibfnamefont{J.}~\bibnamefont{Smolin}}
  (\bibinfo{year}{2017}), \bibinfo{note}{technical report},
  \urlprefix\url{https://dal.objectstorage.open.softlayer.com/v1/AUTH_039c3bf6e6e54d76b8e66152e2f87877/community-documents/quatnum-volumehp08co1vbo0cc8fr.pdf}.

\bibitem[{\citenamefont{Temme et~al.}(2017)\citenamefont{Temme, Bravyi, and
  Gambetta}}]{error-mitigation}
\bibinfo{author}{\bibfnamefont{K.}~\bibnamefont{Temme}},
  \bibinfo{author}{\bibfnamefont{S.}~\bibnamefont{Bravyi}}, \bibnamefont{and}
  \bibinfo{author}{\bibfnamefont{J.~M.} \bibnamefont{Gambetta}},
  \bibinfo{journal}{Phys. Rev. Lett.} \textbf{\bibinfo{volume}{119}},
  \bibinfo{pages}{180509} (\bibinfo{year}{2017}).

\bibitem[{\citenamefont{Li and Benjamin}(2017)}]{error-mitigation1}
\bibinfo{author}{\bibfnamefont{Y.}~\bibnamefont{Li}} \bibnamefont{and}
  \bibinfo{author}{\bibfnamefont{S.~C.} \bibnamefont{Benjamin}},
  \bibinfo{journal}{Phys. Rev. X} \textbf{\bibinfo{volume}{7}},
  \bibinfo{pages}{021050} (\bibinfo{year}{2017}).

\bibitem[{\citenamefont{Kandala et~al.}()\citenamefont{Kandala, Temme,
  Corcoles, Mezzacapo, Chow, and Gambetta}}]{error-mitigation-exp}
\bibinfo{author}{\bibfnamefont{A.}~\bibnamefont{Kandala}},
  \bibinfo{author}{\bibfnamefont{K.}~\bibnamefont{Temme}},
  \bibinfo{author}{\bibfnamefont{A.~D.} \bibnamefont{Corcoles}},
  \bibinfo{author}{\bibfnamefont{A.}~\bibnamefont{Mezzacapo}},
  \bibinfo{author}{\bibfnamefont{J.~M.} \bibnamefont{Chow}}, \bibnamefont{and}
  \bibinfo{author}{\bibfnamefont{J.~M.} \bibnamefont{Gambetta}},
  \eprint{arXiv:1805.04492}.

\end{thebibliography}
\end{document}